\documentclass[11pt,letterpaper]{article}

\usepackage{subcaption}

\usepackage{url}
\usepackage{fullpage}
\usepackage{amsmath}
\usepackage{algorithm}
\usepackage[noend]{algpseudocode}
\usepackage{amssymb}
\usepackage{amsthm}
\usepackage{color}
\usepackage{tikz}
\usetikzlibrary{arrows,shapes}
\usetikzlibrary{matrix}
\usepackage{array}
\usepackage{xy}
\usepackage{setspace}
\usepackage{multicol}
\usepackage{algorithm}
\usepackage{todonotes}

 \newcommand{\E}{\mathbb{E}}

\newcommand{\dtw}{\operatorname{dtw}}
\newcommand{\ed}{\operatorname{ed}}

\newcommand{\xrun}{r_x} \newcommand{\yrun}{r_y}
\newcommand{\yrunlength}{l_y} \newcommand{\yprevrunlength}{p_y}
\newcommand{\xrunlength}{l_x} \newcommand{\yoffset}{o_y}
\newcommand{\xoffset}{o_x}
\newcommand{\yrunoffset}{o_y}
\newcommand{\xrunoffset}{o_x}
\newcommand{\tipToMiddle}{\operatorname{SP}}
\newcommand{\rundiff}{d}

\newtheorem{thm}{Theorem}[section]
\newtheorem{defn}[thm]{Definition}
\newtheorem{lem}[thm]{Lemma}
\newtheorem{prop}[thm]{Proposition}

\newtheorem{cor}[thm]{Corollary}

\theoremstyle{remark}

\newtheorem{ex}[thm]{Example}

\theoremstyle{remark}
\newtheorem{observation}[thm]{Observation}

\usepackage{microtype}

\title{Dynamic Time Warping in Strongly Subquadratic Time: Algorithms
  for the Low-Distance Regime and Approximate Evaluation}

\author{William Kuszmaul \\ Massachusetts Institute of Technology, Cambridge, USA. \\ \textit{kuszmaul@mit.edu}.}

\begin{document}

\maketitle
\begin{abstract}

Dynamic time warping distance (DTW) is a widely used distance measure between time series, with applications in areas such as speech recognition and bioinformatics. The best known algorithms for computing DTW run in near quadratic time, and conditional lower bounds prohibit the existence of significantly faster algorithms.

The lower bounds do not prevent a faster algorithm for the important special case in which the DTW is small, however. For an arbitrary metric space $\Sigma$ with distances normalized so that the smallest non-zero distance is one, we present an algorithm which computes $\dtw(x, y)$ for two strings $x$ and $y$ over $\Sigma$ in time $O(n \cdot \dtw(x, y))$. When $\dtw(x, y)$ is small, this represents a significant speedup over the standard quadratic-time algorithm.

Using our low-distance regime algorithm as a building block, we also present an approximation algorithm which computes $\dtw(x, y)$ within a factor of $O(n^\epsilon)$ in time $\tilde{O}(n^{2 - \epsilon})$ for $0 < \epsilon < 1$. The algorithm allows for the strings $x$ and $y$ to be taken over an arbitrary well-separated tree metric with logarithmic depth and at most exponential aspect ratio. Notably, any polynomial-size metric space can be efficiently embedded into such a tree metric with logarithmic expected distortion. Extending our techniques further, we also obtain the first approximation algorithm for edit distance to work with characters taken from an \emph{arbitrary metric space}, providing an $n^\epsilon$-approximation in time $\tilde{O}(n^{2 - \epsilon})$, with high probability.

Finally, we turn our attention to the relationship between edit distance and dynamic time warping distance. We prove a reduction from computing edit distance over an arbitrary metric space to computing DTW over the same metric space, except with an added null character (whose distance to a letter $l$ is defined to be the edit-distance insertion cost of $l$). Applying our reduction to a conditional lower bound of Bringmann and K\"unnemann pertaining to edit distance over $\{0, 1\}$, we obtain a conditional lower bound for computing DTW over a three letter alphabet (with distances of zero and one). This improves on a previous result of Abboud, Backurs, and Williams, who gave a conditional lower bound for DTW over an alphabet of size five.

With a similar approach, we also prove a reduction from computing edit distance (over generalized Hamming Space) to computing longest-common-subsequence length (LCS) over an alphabet with an added null character. Surprisingly, this means that one can recover conditional lower bounds for LCS directly from those for edit distance, which was not previously thought to be the case.

\end{abstract}

\section{Introduction}

Dynamic Time Warping distance (DTW) is a widely used distance measure
between time series. DTW is particularly flexible in dealing with
temporal sequences that vary in speed.  To measure the distance
between two sequences, portions of each sequence are allowed to be
warped (meaning that a character may be replaced with multiple
consecutive copies of itself), and then the warped sequences are
compared by summing the distances between corresponding pairs of
characters. DTW's many applications include phone authentication
\cite{dtwapp1}, signature verification \cite{dtwapp2}, speech
recognition \cite{dtwapp3}, bioinformatics \cite{dtwapp4}, cardiac
medicine \cite{dtwapp5}, and song identification \cite{dtwapp6}.

The textbook dynamic-programming algorithm for DTW runs in time
$O(n^2)$, which can be prohibitively slow for large inputs. Moreover,
conditional lower bounds \cite{DTWhard, DTWhard2} prohibit the
existence of a strongly subquadratic-time algorithm\footnote{An
  algorithm is said to run in strongly subquadratic time if it runs in
  time $O(n^{2- \epsilon})$ for some constant $\epsilon > 0$. Although
  strongly subquadratic time algorithms are prohibited by 
  conditional lower bounds, runtime improvements by
  subpolynomial factors are not. Such improvements have been
  achieved \cite{DTWsubquadratic}.}, unless the Strong Exponential
Time Hypothesis is false.

The difficulty of computing DTW directly has motivated the development
of fast heuristics \cite{dtwband,kp99,kp00,k02,BUWK15,PFWNCK16} which
typically lack provable guarantees.

On the theoretical side, researchers have considered DTW in the
contexts of locality sensitive hashing \cite{ds17} and nearest
neighbor search \cite{ep17}, but very little additional progress has
been made on the general problem in which one is given two strings $x$
and $y$ with characters from a metric space $\Sigma$, and one wishes to
compute (or approximate) $\dtw(x, y)$.

To see the spectrum of results one might aim for, it is helpful to
consider edit distance, another string-similarity measure, defined to
be the minimum number of insertions, deletions, and substitutions
needed to get between two strings. Like DTW, edit distance can
be computed in time $O(n^2)$ using dynamic programming
\cite{WagnerF74, needleman1970general, vintzyuk1968}, and conditional
lower bounds suggest that no algorithm can do significantly better
\cite{backurs2015edit, DTWhard}. This has led to researchers focusing
on specialized versions of the problem, especially in two important
directions:
\begin{itemize}
\item \textbf{Low-Distance Regime Algorithms: } In the special case
  where the edit distance between two strings is small (less than
  $\sqrt{n}$), the algorithm of Landau, Myers and Schmidt
  \cite{landau1998incremental} can be used to compute the exact
  distance in time $O(n)$. In general, the algorithm runs in time $O(n
  + \ed(x, y)^2)$. Significant effort has also been made to design
  variants of the algorithm which exhibit small constant overhead in
  practice \cite{chakraborty2016streaming2}.
\item \textbf{Approximation Algorithms: } Andoni, Krauthgamer and Onak
  introduced an algorithm estimating edit distance within a factor of
  $(\log n)^{O(1/\varepsilon)}$ in time $O(n^{1 + \varepsilon})$
  \cite{ApproxPolyLog}, culminating a long line of research on
  approximation algorithms that run in close to linear time
  \cite{ApproxSubPolyDistortion,DimensionReduction,bar2004approximating,CharikarGe18}. Recently,
  Chakraborty et al. gave the first strongly subquadratic algorithm to
  achieve a \emph{constant} approximation, running time $O(n^{12/7})$
  \cite{approx18}.
\end{itemize}

This paper presents the first theoretical results in
these directions for DTW.

\paragraph*{A Low-Distance Regime Algorithm for DTW (Section \ref{seclowdistance}) }
We present the first algorithm for computing DTW in the low-distance
regime. Our algorithm computes $\dtw(x, y)$ in time $O(n \cdot \dtw(x,
y))$ for strings $x$ and $y$ with characters taken from an arbitrary
metric space in which the minimum non-zero distance is one. The key
step in our algorithm is the design of a new dynamic-programming
algorithm for DTW, which lends itself especially well to the
low-distance setting.

Our dynamic program relies on a recursive structure in which the two
strings $x$ and $y$ are treated asymmetrically within each subproblem:
One of the strings is considered as a sequence of letters, while the
other string is considered as a sequence of runs (of equal
letters). The subproblems build on one-another in a way so that at
appropriate points in the recursion, we toggle the role that the two
strings play. The asymmetric treatment of the strings limits the
number of subproblems that can have return-values less than any given
threshold $K$ to $O(nK)$, allowing for a fast algorithm in the
low-distance setting.


We remark that the requirement of having the smallest distance between
distinct characters be $1$ is necessary for the low-distance regime
algorithm to be feasible, since otherwise distances can simply be
scaled down to make every DTW instance be low-distance.

\paragraph*{Approximating DTW Over Well Separated Tree Metrics (Section \ref{secdtwapprox}) }
We design the first approximation algorithm for DTW to run in strongly
subquadratic time. Our algorithm computes $\dtw(x, y)$ within an
$n^\epsilon$-approximation in time $\tilde{O}(n^{2 - \epsilon})$. The
algorithm allows for the strings $x$ and $y$ have characters taken
from an arbitrary well-separated tree metric of logarithmic depth and
at most exponential aspect ratio.\footnote{The \emph{aspect} ratio of
  a metric space is the ratio between the largest and smallest
  non-zero distances in the space.} These metric spaces are universal
in the sense that any finite metric space $M$ of polynomial size can
be efficiently embedded into a well-separated tree metric with
expected distortion $O(\log |M|)$ and logarithmic depth \cite{trees,
  bansal2011}.

An important consequence of our approximation algorithm is for the
special case of DTW over the reals. Exploiting a folklore embedding
from $\mathbb{R}$ to a well-separated tree metric metric, we are able
to obtain with high probability an $O(n^{\epsilon})$-approximation for
$\dtw(x, y)$ in time $\tilde{O}(n^{2 - \epsilon})$, for any strings
$x$ and $y$ of length at most $n$ over a subset of the reals with a
polynomial aspect ratio.

In the special case of DTW over the reals, previous work has been done
to find approximation algorithms under certain geometric assumptions
about the inputs $x$ and $y$ \cite{agarwal2015approximating,
  ying2016simple}. To the best of our knowledge, our approximation
algorithm is the first to not rely on any such assumptions.

It is interesting to note that our results on low-distance regime and
approximation algorithms for DTW have bounds very similar to the
earliest results for edit distance in the same directions. Indeed, the
first algorithm to compute edit distance in the low-distance regime
\cite{firstlowdist} exploited properties of a (now standard)
dynamic-programming algorithm in order to compute $\ed(x, y)$ in time
$O(n \cdot \ed(x, y))$. This implicitly resulted in the first
approximation algorithm for edit distance, allowing one to compute an
$O(n^{\epsilon})$-approximation in time $O(n^{2 - \epsilon})$. Until
the work of \cite{ApproxSubPolyDistortion} and \cite{ApproxPolyLog},
which culminated in an algorithm with a polylogarithmic approximation
ratio, the best known approximation ratio for edit distance remained
polynomial for roughly twenty years \cite{landau1998incremental,
  bar2004approximating,DimensionReduction}.

The $\tilde{O}(n^{2 - \epsilon})$-time $O(n^{\epsilon})$-approximation
tradeoff is also the current state-of-the-art for another related
distance measurement known as Fr\'echet distance \cite{BringmannMu15},
and is achieved using an algorithm that differs significantly from its
edit-distance and DTW counterparts.

\paragraph*{Reduction from Edit Distance to DTW (Section \ref{secreduction})}
We show that the similarity between our results for DTW and the
earliest such results for edit distance is not coincidental. In
particular, we prove a simple reduction from computing edit distance
over an arbitrary metric space to computing DTW over the same metric
space (with an added null character). Consequently, any algorithmic
result for computing DTW in the low-distance regime or
approximating DTW immediately implies the analogous result for edit
distance. The opposite direction is true for lower bounds. For
example, the conditional lower bound of Bringmann and K\"unnemann
\cite{DTWhard}, which applies to edit distance over the alphabet $\{0,
1\}$, now immediately implies a conditional lower bound for DTW over
an alphabet of size three (in which characters are compared with
distances zero and one). This resolves a direction of work posed by
Abboud, Backurs, and Williams \cite{DTWhard2}, who gave a conditional
lower bound for DTW over an alphabet of size five, and noted that if
one could prove the same lower bound for an alphabet of size three,
then the runtime complexity of DTW over generalized Hamming space
would be settled (modulo the Strong Exponential Time
Hypothesis). Indeed, it is known that over an alphabet of size two,
DTW can be computed in strongly subquadratic time \cite{DTWhard2}.

Using a similar approach we also prove a simple reduction from
computing edit distance (over generalized Hamming space) to computing
the longest-common-subsequence length (LCS) between two strings. Thus
conditional lower bounds for computing edit distance directly imply
conditional lower bounds for computing LCS (over an alphabet with one
additional character). This was not previously though to be the
case. Indeed, the first known conditional lower bounds for LCS came
after those for edit distance \cite{AbboudBa15, backurs2015edit}, and
it was noted by Abboud et al. \cite{AbboudBa15} that \emph{``A simple
  observation is that the computation of the LCS is equivalent to the
  computation of the Edit-Distance when only deletions and insertions
  are allowed, but no substitutions. Thus, intuitively, LCS seems like
  an easier version of Edit Distance, since a solution has fewer
  degrees of freedom, and the lower bound for Edit-Distance does not
  immediately imply any hardness for LCS.''} Our reduction violates
this intuition by showing that edit distance without substitutions can
be used to efficiently simulate edit distance without substitutions.
In addition to presenting a reduction from edit distance to LCS, we
show that no similar reduction can exist in the other direction.

We remark that our reduction from edit distance to LCS was also
previously briefly described in Chapter 6.1 of \cite{reduction_past}, a fact of
which we only became aware after the completion of this manuscript.

\paragraph*{Approximating Edit Distance Over an Arbitrary Metric (Section \ref{seceditapprox})}
The aforementioned results for approximating edit distance
\cite{firstlowdist, bar2004approximating, DimensionReduction,
  landau1998incremental, ApproxSubPolyDistortion, ApproxPolyLog, approx18, Kuszmaul19}
consider only the case in which insertion, deletion, and substitution
costs are all constant. To the best of our knowledge, no approximation
algorithm is known for the more general case in which characters are
taken from an arbitrary metric space and edit costs are assigned based
on metric distances between characters. This variant of edit distance
is sometimes referred to as \emph{general edit distance}
\cite{navarro2001guided}. The study of general edit distance dates
back to the first papers on edit distance \cite{WagnerF74,
  ukkonen}, and allowing for nonuniform costs is important in many
applications, including in computational biology
\cite{jiang2002general}.

We present an approximation algorithm for edit distance over an
an arbitrary metric. Our algorithm runs in time $\tilde{O}(n^{2 -
  \epsilon})$ and computes an $O(n^{\epsilon})$-approximation for
$\ed(x, y)$ with high probability. Note that for the case where
characters are taken from a well-separated tree metric with
logarithmic depth and at most exponential aspect ratio, the result
already follows from our approximation algorithm for DTW, and our
reduction from edit distance to DTW. The approach taken in Section
\ref{seceditapprox} is particularly interesting in that it places no
restrictions on the underlying metric space.

Both our approximation algorithm for DTW and our approximation
algorithm for edit distance exhibit relatively weak
runtime/approximation tradeoffs. To the best of our knowledge,
however, they are the first such algorithms to run in
strongly subquadratic time.

\section{Preliminaries}

In this section, we present preliminary definitions and background on
dynamic time warping distance (DTW) and edit distance.

\paragraph*{Dynamic Time Warping Distance} For a metric space $\Sigma$,
the dynamic time warping distance (DTW) between two strings $x, y \in
\Sigma^n$ is a natural measure of similarity between the strings. 

Before fully defining DTW, we first introduce the notion of
an \emph{expansion} of a string.

\begin{defn}
The \emph{runs} of a string $x \in \Sigma^n$ are the maximal
subsequences of consecutive letters with the same value. One can
\emph{extend} a run by replacing it with a longer run of the same
letter. An \emph{expansion} of the string $x$ is any string which can
be obtained from $x$ by extending runs.
\end{defn}

As an example, consider $x = aaaccbbd$. Then the runs of $x$ are
$aaa$, $cc$, $bb$, and $d$. The string $\overline{x} = aaacccccbbdd$
is an expansion of $x$ and extends the runs containing $c$
and $d$.

Using the terminology of expansions, we now define DTW.

\begin{defn}
 Consider strings $x$ and $y$ of length $n$ over a metric $(\Sigma,
 d)$. A \emph{correspondence} $(\overline{x}, \overline{y})$ between
 $x$ and $y$ is a pair of equal-length expansions $\overline{x}$ of
 $x$ and $\overline{y}$ of $y$. The \emph{cost} of a correspondence is
 given by $\sum_i d(\overline{x}_i, \overline{y}_i).$

 The \emph{dynamic time warping distance} $\dtw(x, y)$ is defined to
 be the minimum cost of a correspondence between $x$ and $y$.
\end{defn}

When referring to a run $r$ in one of $x$ or $y$, and when talking
about a correspondence $(\overline{x}, \overline{y})$, we will often
use $r$ to implicitly refer to the extended run corresponding with $r$
in the correspondence. Whether we are referring to the original run or
the extended version of the run should be clear from context.

Note that any minimum-length optimal correspondence between strings
$x, y \in \Sigma^{n}$ will be of length at most $2n$. This is because
if a run $r_1$ in $x$ overlaps a run $r_2$ in $y$ in the
correspondence, then we may assume without loss of generality that at
most one of the two runs is extended by the
correspondence. (Otherwise, we could un-extend each run by one and
arrive at a shorter correspondence with no added cost.)


\paragraph*{Edit Distance Over an Arbitrary Metric}
The \emph{simple edit distance} between two strings $x$ and $y$ is the
minimum number of insertions, deletions, and substitutions needed to
transform $x$ into $y$. In this paper we will mostly focus on a more
general variant of edit distance, in which characters are taken from an
arbitrary metric:

\begin{defn}
Let $x$ and $y$ be strings over an alphabet $\Sigma$, where $(\Sigma
\cup \{\emptyset\}, d)$ is a metric space. We say that the
\emph{magnitude} $|l|$ of a letter $l \in \Sigma$ is $d(\emptyset,
l)$.  We define the \emph{edit distance} between $x$ and $y$ to be the
minimum cost of a sequence of edits from $x$ to $y$, where the
insertion or deletion of a letter $l$ costs $d(\emptyset, l)$, and the
substitution of a letter $l$ to a letter $l'$ costs $d(l, l')$.
\end{defn}

\section{Computing DTW in the Low-Distance Regime}\label{seclowdistance}
In this section, we present a low-distance regime algorithm for DTW
(with characters from an arbitrary metric in which all non-zero
distances are at least one). Given that $\dtw(x, y)$ is bounded above
by a parameter $K$, our algorithm can compute $\dtw(x, y)$ in time
$O(nK)$. Moreover, if $\dtw(x, y) > K$, then the algorithm will
conclude as much. Consequently, by doubling our guess for $K$
repeatedly, one can compute $\dtw(x, y)$ in time $O(n \cdot \dtw(x,
y))$.

Consider $x$ and $y$ of length $n$ with characters taken from a metric
space $\Sigma$ in which all non-zero distances are at least one. In
the textbook dynamic program for DTW \cite{algorithmdesign}, each pair
of indices $i, j \in [n]$ represents a subproblem $T(i, j)$ whose
value is $\dtw(x[1:i], y[1: j])$. Since $T(i, j)$ can be determined
using $T(i - 1, j), T(i, j - 1), T(i - 1, j - 1)$, and knowledge of
$x_i$ and $y_j$, this leads to an $O(n^2)$ algorithm for DTW. A common
heuristic in practice is to construct only a small band around the
main diagonal of the dynamic programming grid; by computing only
entries $T(i, j)$ with $|i - j| \le 2K$, and treating other
subproblems as having infinite return values, one can obtain a correct
computation for DTW as long as there is an optimal correspondence
which matches only letters which are within $K$ of each other in
position. This heuristic is known as the Sakoe-Chiba Band
\cite{dtwband} and is employed, for example, in the commonly used
library of Giorgino \cite{giorgino}.

The Sakoe-Chiba Band heuristic can perform badly even when $\dtw(x,
y)$ is very small, however. Consider $x = abbb\cdots b$ and $y = aaa
\cdots ab$. Although $\dtw(x, y) = 0$, if we restrict ourselves to
matching letters within $K$ positions of each other for some small
$K$, then the resulting correspondence will cost $\Omega(n)$.

In order to obtain an algorithm which performs well in the
low-distance regime, we introduce a new dynamic program for DTW. The
new dynamic program treats $x$ and $y$ asymmetrically within each
subproblem. Loosely speaking, for indices $i$ and $j$, there are
two subproblems $\tipToMiddle(x, y, i, j)$ and $\tipToMiddle(y, x, i,
j)$. The first of these subproblems evaluates to the DTW between the
first $i$ runs of $x$ and the first $j$ letters of $y$, with the added
condition that the final run of $y[1: j]$ is not extended. The second
of the subproblems is analogously defined as the DTW between the
first $i$ runs of $y$ and the first $j$ letters of $x$ with the added
condition that the final run of $x[1: j]$ is not extended.

The recursion connecting the new subproblems is somewhat more
intricate than for the textbook dynamic program. By matching the
$i$-th run with the $j$-th letter, however, we limit the number of
subproblems which can evaluate to less than $K$. In particular, if the
$j$-th letter of $y$ is in $y$'s $t$-th run, then any correspondence
which matches the $i$-th run of $x$ to the $j$-th letter of $y$ must
cost at least $\Omega(|i - t|)$. (This is formally shown in Appendix \ref{seclowdistancefull}.) Thus for a given $j$, there are only
$O(K)$ options for $i$ such that $\tipToMiddle(x, y, i, j)$ can
possibly be at most $K$, and similarly for $\tipToMiddle(y, x, i,
j)$. Since we are interested in the case of $\dtw(x, y) \le K$, we can
restrict ourselves to the $O(nK)$ subproblems which have the potential
to evaluate to at most $O(K)$. Notice that, in fact, our algorithm
will work even when $\dtw(x, y) > K$ as long as there is an optimal
correspondence between $x$ and $y$ which only matches letters from $x$
from the $\xrun$-th run with letters from $y$ from the $\yrun$-th run
if $|\xrun - \yrun| \le O(K)$.

Formally we define our recursive problems in a manner slightly
different from that described above. Let $x$ and $y$ be strings of
length at most $n$ and let $K$ be a parameter which we assume is
greater than $\dtw(x, y)$. Our subproblems will be the form
$\tipToMiddle(x, y, \xrun, \yrun, \yrunoffset)$, which is defined as
follows. Let $x'$ consist of the first $\xrun$ runs of $x$ and $y'$
consist of the first $\yrun$ runs of $y$ until the $\yrunoffset$-th
letter in the $\yrun$-th run. Then $\tipToMiddle(x, y, \xrun, \yrun,
\yrunoffset)$ is the value of the optimal correspondence between $x'$
and $y'$ such that the $\yrun$-th run in $y'$ is not
extended.\footnote{If $\yrunoffset = 0$, then the $\yrun$-th run in
  $y'$ is empty and thus trivially cannot be extended.} If no such
correspondence exists (which can only happen if $\yrun \le 1$ or
$\xrun = 0$), then the value of the subproblem is $\infty$. Note that
we allow $\xrun, \yrun, \yrunoffset$ to be zero, and if $\yrun$ is
zero, then $\yrunoffset$ must be zero as well. We also consider the
symmetrically defined subproblems of the form $\tipToMiddle(y, x,
\yrun, \xrun, \xoffset)$. We will focus on the subproblems of the
first types, implicitly treating subproblems of the second type
symmetrically.

\begin{ex}
  Suppose characters are taken from generalized Hamming space, with
  distances of 0 and 1. The subproblem $\tipToMiddle(efabbccccd,
  ffaabcccddd, 5, 4, 2)$ takes the value of the optimal correspondence
  between $efabbcccc$ and $ffaabcc$ such that the final $cc$ run in
  the latter is not extended. The subproblem's value turns out to be
  3, due to the correspondence:
  \begin{center}
  \begin{tabular}{l l l l l l l l l l l}
    e&f&a&a&b&b&c&c&c&c&  \\
    f&f&a&a&b&b&b&b&c&c&.  \\
  \end{tabular}
  \end{center}
\end{ex}

The next lemma presents the key recursive relationship
between subproblems. The lemma focuses on the case where $\xrun, \yrun,
\yrunoffset \ge 1$.

\begin{lem}
  Suppose that $\xrun$ and $\yrun$ are both between $1$ and the number
  of runs in $x$ and $y$ respectively; and that $\yrunoffset$ is
  between $1$ and the length of the $\yrun$-th run in $y$. Let
  $\xrunlength$ be the length of the $\xrun$-th run in $x$ and
  $\yrunlength$ be the length of the $\yrun$-th run in $y$. Let
  $\rundiff$ be the distance between the letter populating the
  $\xrun$-th run in $x$ and the letter populating the $\yrun$-th run
  in $y$. Then $\tipToMiddle(x, y, \xrun, \yrun, \yrunoffset)$ is
  given by
  \[\begin{cases}
    \min\left(\tipToMiddle(x, y, \xrun, \yrun, \yrunoffset - 1) + \rundiff,  \tipToMiddle(x, y, \xrun - 1, \yrun, \yrunoffset - \xrunlength) + \rundiff \cdot \xrunlength\right) & \text{ if }\xrunlength \le \yrunoffset \\
    \min\left(\tipToMiddle(x, y, \xrun, \yrun, \yrunoffset - 1) + \rundiff,  \tipToMiddle(y, x, \yrun - 1, \xrun, \xrunlength - \yrunoffset) + \rundiff \cdot \yrunoffset\right) & \text{ if }\xrunlength > \yrunoffset. \\
  \end{cases}\]
    \label{lemrecursion10}
\end{lem}
\begin{proof}
  Consider a minimum-cost correspondence $A$ between the first $\xrun$ runs
  of $x$ and the portion of $y$ up until the $\yrunoffset$-th letter
  in the $\yrun$-th run, such that the $\yrun$-th run in $y$ is not
  extended.

  If the $\xrun$-th run in $A$ is extended, then the cost of $A$ will
  be $\tipToMiddle(x, y, \xrun, \yrun, \yrunoffset - 1) + \rundiff$. If the $\xrun$-th run in $A$ is not extended, then we consider two
  cases.

  In the first case, $\xrunlength \le \yrunoffset$. In this case, the
  entirety of the $\xrun$-th run of $x$ is engulfed by the $\yrun$-th
  run of $y$ in the correspondence $A$. Since the $\xrun$-th run is
  not extended, the cost of the overlap is $\xrunlength \cdot d$. Thus
  the cost of $A$ must be
  $\tipToMiddle(x, y, \xrun - 1, \yrun, \yrunoffset - \xrunlength) +
  \rundiff \cdot \xrunlength.$

  Moreover, since $A$ is minimum-cost, as long as
  $\xrunlength \le \yrunoffset$, the cost of $A$ is at most the above
  expression, regardless of whether the $\xrun$-th run in $x$ is
  extended in $A$.

  In the second case, $\xrunlength > \yrunoffset$. In this case, the
  first $\yoffset$ letters in the $\yrun$-th run of $y$ all overlap
  the $\xrun$-th run of $x$ in $A$. Since the $\yrun$-th run is not
  extended, the cost of the overlap is $d \cdot \yoffset$. Thus, since
  the $\xrun$-th run in $x$ is also not extended in $A$, the cost of
  $A$ must be
  $\tipToMiddle(y, x, \yrun - 1, \xrun, \xrunlength - \yrunoffset) + \rundiff \cdot
  \yoffset.$
  Moreover, since $A$ is minimal, as long as $\xrunlength
  > \yrunoffset$, the cost of $A$ is at most the above expression,
  regardless of whether the $\xrun$-th run in $x$ is extended.
\end{proof}

The above lemma handles cases where $\xrun, \yrun, \yrunoffset >
0$. In the case where $\xrun > 0$, $\yrun > 0$, and $\yrunoffset = 0$,
$\tipToMiddle(x, y, \xrun, \yrun, \yrunoffset)$ is just the dynamic
time warping distance between the first $\xrun$ runs of $x$ and the
first $\yrun - 1$ runs of $y$, given by
$\min\left(\tipToMiddle(x, y, \xrun, \yrun - 1, t_1), \tipToMiddle(y, x, \yrun - 1, \xrun, t_2) \right),$
where $t_1$ is the length of the $(\yrun - 1)$-th run in $y$ and
$t_2$ is the length of the $\xrun$-th run in $x$. The remaining cases are edge-cases with $\tipToMiddle(x, y, \xrun, \yrun,
\yrunoffset) \in \{0, \infty\}$. (See Appendix
\ref{seclowdistancefull}.)

One can show that any correspondence $A$ in which a letter from the
$\xrun$-th run of $x$ is matched with a letter from the $\yrun$-th run
of $y$ must contain must contain at least $\frac{|\xrun - \yrun| -
  1}{2}$ instances of unequal letters being matched; we prove this in
Appendix \ref{seclowdistancefull}. It follows that if $\dtw(x, y) \le
K$, then we can limit ourselves to subproblems in which $|\xrun -
\yrun| \le O(K)$. For each of the $n$ options of $(\yrun,
\yrunoffset)$, there are only $O(K)$ options of $\xrun$ that must be
considered. This limits the total number of subproblems to
$O(nK)$. The resulting dynamic program yields the following theorem:

\begin{thm}
Let $x$ and $y$ be strings of length $n$ taken from a metric space
$\Sigma$ with minimum non-zero distance at least one, and let $K$ be
parameter such that $\dtw(x, y) \le K$. Then there exists a dynamic
program for computing $\dtw(x, y)$ in time $O(nK)$. Moreover, if
$\dtw(x, y) > K$, then the dynamic program will return a value greater
than $K$.
\label{thmdp}
\end{thm}

By repeatedly doubling one's guess for $K$ until the computed value of
$\dtw(x, y)$ evaluates to less than $K$, one can therefore compute
$\dtw(x, y)$ in time $O(n \cdot \dtw(x, y))$.


\section{Approximating DTW Over Well-Separated Tree Metrics}\label{secdtwapprox}

In this section, we present an $\tilde{O}(n^{2 - \epsilon})$-time
$O(n^{\epsilon})$-approximation algorithm for DTW over a
well-separated tree metric with logarithmic depth. We begin by
presenting a brief background on well-separated tree metrics.

\begin{defn}
  Consider a tree $T$ whose vertices form an alphabet $\Sigma$, and
  whose edges have positive weights. $T$ is said to be a
  \emph{well-separated tree metric} if every root-to-leaf path
  consists of edges ordered by nonincreasing weight. The
  \emph{distance} between two nodes $u, v \in \Sigma$ is defined as
  the maximum weight of any edge in the shortest path from $u$ to $v$.

\end{defn}

Well-separated tree metrics are universal in the sense that any metric
$\Sigma$ can be efficiently embedded (in time $O(|\Sigma|^2)$) into a
well-separated tree metric $T$ with expected distortion $O(\log
|\Sigma|)$ \cite{trees}. Moreover, the tree metric may be
  made to have logarithmic depth using Theorem 8 of
  \cite{bansal2011}. For strings $x, y \in \Sigma^n$, let $\dtw_T(x,
y)$ denote the dynamic time warping distance after embedding $\Sigma$
into $T$. Then the tree-metric embedding guarantees that $\dtw(x, y)
\le \dtw_T(x, y)$ and that $\E[\dtw_T(x, y)] \le O(\log n) \cdot \dtw(x, y)$. (The
latter fact is slightly nontrivial and is further explained in
Appendix \ref{secdtwapproxfull}.)

It follows that any approximation algorithm for DTW over
well-separated tree metrics will immediately yield an approximation
algorithm over an arbitrary polynomial-size metric $\Sigma$, with two
caveats: the new algorithm will have its multiplicative error
increased by $O(\log n)$; and $O(\log n)$ instances of $\Sigma$
embedded into a well-separated tree metric must be precomputed for use
by the algorithm (requiring, in general, $O(|\Sigma|^2 \log n)$
preprocessing time). In particular, given $O(\log n)$ tree embeddings of
$\Sigma$, $T_1, \ldots, T_{O(\log n)}$, with high probability $\min_i
\left(\dtw_{T_i}(x, y)\right)$ will be within a logarithmic factor of
$\dtw(x, y)$.

The remainder of the section will be devoted to designing an
approximation algorithm for DTW over a well-separated tree
metric. We will prove the following theorem:

\begin{thm}
  Consider $0 < \epsilon < 1$. Suppose that $\Sigma$ is a
  well-separated tree metric of polynomial size and at most
  logarithmic depth. Moreover, suppose that the aspect ratio of
  $\Sigma$ is at most exponential in $n$ (i.e., the ratio between the
  largest distance and the smallest non-zero distance). Then in time
  $\tilde{O}(n^{2 - \epsilon})$ we can obtain an
  $O(n^{\epsilon})$-approximation for $\dtw(x, y)$ for any $x, y \in
  \Sigma^n$.
  \label{thmdtwapproxtree0}
\end{thm}

An important consequence of the theorem occurs for DTW over the
reals. When $\Sigma$ is an $O(n)$-point subset of the reals with a
polynomial aspect ratio, there exists an $O(n\log n)$-time embedding
with $O(\log n)$ expected distortion from $\Sigma$ to a well-separated
tree metric of size $O(n)$ with logarithmic depth. (See Appendix
\ref{secdtwapproxfull}). This gives the following corollary:

\begin{cor}
  Consider $0 < \epsilon < 1$. Suppose that $\Sigma = [0, n^c] \cap
  \mathbb{Z}$ for some constant $c$. Then in time $\tilde{O}(n^{2 -
    \epsilon})$ we can obtain an $O(n^{\epsilon})$-approximation for
  $\dtw(x, y)$ with high probability for any $x, y \in \Sigma^n$.
  \label{corsimple0}
\end{cor}

In proving Theorem \ref{thmdtwapproxtree0}, our approximation algorithm
will take advantage of what we refer to as the
\emph{$r$-simplification} of a string over a well-separated tree
metric.
\begin{defn}
Let $T$ be a well-separated tree metric whose nodes form an alphabet $\Sigma$. For a
string $x \in \Sigma^n$, and for any $r \ge 1$, the
\emph{$r$-simplification} $s_r(x)$ is constructed by replacing
each letter $l \in x$ with its highest ancestor $l'$ in $T$ that can
be reached from $l$ using only edges of weight at most $r / 4$.
\end{defn}

Our approximation algorithm will apply the low-distance regime
algorithm from the previous section to $s_r(x)$ and $s_r(y)$ for
various $r$ in order to extract information about $\dtw(x, y)$. Notice
that using our low-distance regime algorithm for DTW, we get the
following useful lemma for free:

\begin{lem}
  Consider $0 < \epsilon < 1$. Suppose that for all pairs $l_1, l_2$
  of distinct letters in $\Sigma$, $d(l_1, l_2) \ge \gamma$. Then for
  $x, y \in \Sigma^n$ there is an $O(n^{2 - \epsilon})$ time algorithm
  which either computes $\dtw(x, y)$ exactly, or concludes that $\dtw(x,
  y) > \gamma n^{1 - \epsilon}$.
  \label{lemdiagonalalgdtw0}
\end{lem}

The next lemma states three important properties of
$r$-simplifications. We remark that the same lemma appears in our
concurrent work on the communication complexity of DTW, in which we
use the lemma in designing an efficient
one-way communication protocol \cite{dtwcomm}.

\begin{lem}
  Let $T$ be a well-separated tree metric with distance function $d$
  and whose nodes form the alphabet $\Sigma$. Consider strings $x$ and
  $y$ in $\Sigma^{n}$.

  Then the following three properties of $s_r(x)$ and $s_r(y)$ hold:
  \begin{itemize}
  \item For every letter $l_1 \in s_r(x)$ and every letter $l_2 \in
    s_r(y)$, if $l_1 \neq l_2$, then $d(l_1, l_2) > r / 4$.
  \item For all $\alpha$, if $\dtw(x, y) \le nr / \alpha$ then $\dtw(s_r(x), s_r(y)) \le nr / \alpha$.
  \item If $\dtw(x, y) > nr$, then $\dtw(s_r(x), s_r(y)) > nr/2$. 
    \end{itemize}
  \label{lemthreeprops0}
\end{lem}

The first and second parts of Lemma \ref{lemthreeprops0} are
straightforward from the definitions of $s_r(x)$ and
$s_r(y)$. The third part follows from the observation
that a correspondence $C$ between $x$ and $y$ can cost at most $|C|
\cdot \frac{r}{4}$ more than the corresponding correspondence between
$s_r(x)$ and $s_r(y)$, where $|C|$ denotes the length of the
correspondence. Since there exists an optimal correspondence between
$s_r(x)$ and $s_r(y)$ of length no more than $2n$, it follows that
$\dtw(x, y) \le \dtw(s_r(x), s_r(y)) + nr/2$, which implies the third
part of the lemma.

A full proof of Lemma \ref{lemthreeprops0} appears in Appendix
\ref{secdtwapproxfull}. Next we prove Theorem \ref{thmdtwapproxtree0}.

\begin{proof}[Proof of Theorem \ref{thmdtwapproxtree0}]
Without loss of generality, the minimum non-zero distance in $\Sigma$ is 1 and
the largest distance is some value $m$, which is at most exponential
in $n$.

We begin by defining the \emph{$(r, n^{\epsilon})$-DTW gap} problem
for $r \ge 1$, in which for two strings $x$ and $y$ a return value of
0 indicates that $\dtw(x, y) < nr$ and a return value of 1 indicates
that $\dtw(x, y) \ge n^{1 - \epsilon}r$. By Lemma \ref{lemthreeprops0},
in order to solve the $(r, n^{\epsilon})$-DTW gap problem for $x$ and
$y$, it suffices to determine whether $\dtw(s_r(x), s_r(y)) \le n^{1 -
  \epsilon}r$. Moreover, because the minimum distance between distinct
letters in $s_r(x)$ and $s_r(y)$ is at least $r / 4$, this can be done
in time $O(n^{2 - \epsilon} \log n)$ using Lemma
\ref{lemdiagonalalgdtw0}.\footnote{The logarithmic factor comes from
  the fact that evaluating distances between points may take
  logarithmic time in our well-separated tree metric.}

In order to obtain an $n^{\epsilon}$-approximation for $\dtw(x, y)$,
we begin by using Lemma \ref{lemdiagonalalgdtw0} to either determine
$\dtw(x, y)$ or to determine that $\dtw(x, y) \ge n^{1 -
  \epsilon}$. For the rest of the proof, suppose we are in the latter
case, meaning that we know $\dtw(x, y) \ge n^{1 - \epsilon}$.

We will now consider the $(2^i, n^{\epsilon} / 2)$-DTW gap problem for
$i \in \{0, 1, 2, \ldots, \lceil \log m \rceil\}$. (Recall that $m$ is
the largest distance in $\Sigma$.) If the $(2^0, n^{\epsilon} / 2)$-DTW
gap problem returned 0, then we would know that $\dtw(x, y) \le n$,
and thus we could return $n^{1 - \epsilon}$ as an
$n^{\epsilon}$-approximation for $\dtw(x, y)$. Therefore, we need only
consider the case where the $(2^0, n^{\epsilon} / 2)$-DTW gap returns
$1$. Moreover we may assume without computing it that $(2^{ \lceil
  \log m \rceil}, n^{\epsilon}/2)$-DTW gap returns 0 since trivially
$\dtw(x, y)$ cannot exceed $nm$. Because $(2^i, n^{\epsilon} / 2)$-DTW
gap returns 1 for $i = 0$ and returns $0$ for $i = \lceil \log m
\rceil$, there must be some $i$ such that $(2^{i - 1}, n^{\epsilon} /
2)$-DTW gap returns $1$ and $(2^{i}, n^{\epsilon} / 2)$-DTW gap
returns 0. Moreover, we can find such an $i$ by performing a binary
search on $i$ in the range $R = \{0, \ldots, \lceil \log m
\rceil\}$. We begin by computing $(2^i, n^{\epsilon} / 2)$-DTW gap for
$i$ in the middle of the range $R$. If the result is a one, then we
can recurse on the second half of the range; otherwise we recurse on
the first half of the range. Continuing like this, we can find in time
$\tilde{O}(n^{2 - \epsilon} \log \log m) = \tilde{O}(n^{2 -
  \epsilon})$ some value $i$ for which $(2^{i - 1}, n^{\epsilon} /
2)$-DTW gap returns $1$ and $(2^{i}, n^{\epsilon} / 2)$-DTW gap
returns 0. Given such an $i$, we know that $\dtw(x, y) \ge \frac{2^{i
    - 1}n}{n^{\epsilon} / 2} = 2^i n^{1 - \epsilon}$ and that $\dtw(x,
y) \le 2^in$. Thus we can return $2^i n^{1 - \epsilon}$ as an
$n^{\epsilon}$ approximation of $\dtw(x, y)$.
\end{proof}

\section{Reducing Edit Distance to DTW and LCS}\label{secreduction}

In this section we present a simple reduction from edit distance over
an arbitrary metric to DTW over the same metric.
 At the end of the section, we
prove as a corollary a conditional lower bound for DTW over
three-letter Hamming space, prohibiting any algorithm from running in
strongly subquadratic time.

Surprisingly, the \emph{exact same reduction}, although with a
different analysis, can be used to reduce the computation of edit
distance (over generalized Hamming space) to the computation of
longest-common-subsequence length (LCS). Since computing LCS is
equivalent to computing edit distance without substitutions, this
reduction can be interpreted as proving that edit distance without
substitutions can be used to efficiently simulate edit distance with
substitutions, also known as \emph{simple edit distance}.

Recall that for a metric $\Sigma \cup \{\emptyset\}$, we define the
edit distance between two strings $x, y \in \Sigma^n$ such that the
cost of a substitution from a letter $l_1$ to $l_2$ is $d(l_1, l_2)$,
and the cost of a deletion or insertion of a letter $l$ is $d(l,
\emptyset)$. Additionally, define the \emph{simple edit distance}
$\ed_S(x, y)$ to be the edit distance using only insertions and
deletions.

For a string $x \in \Sigma^n$, define the \emph{padded string} $p(x)$
of length $2n + 1$ to be the string $\emptyset x_1 \emptyset x_2
\emptyset x_3 \cdots x_n \emptyset$. In particular, for $i \le 2n +
1$, $p(x)_i = \emptyset$ when $i$ is odd, and $p(x)_i = x_{i / 2}$
when $i$ is even. The following theorem proves that $\dtw(p(x), p(y))
= \ed(x, y)$.

\begin{thm}
  Let $\Sigma \cup \{\emptyset\}$ be a metric. Then for any $x, y \in
  \Sigma^n$,
  $\dtw(p(x), p(y)) = \ed(x, y).$
  \label{thmreduction00}
\end{thm}
\begin{proof}[Proof sketch]
A key observation is that when constructing an optimal correspondence
between $p(x)$ and $p(y)$, one may w.l.o.g. extend only runs
consisting of $\emptyset$ characters. In particular, suppose that one
extends a non-$\emptyset$ character $a$ in $p(x)$ to match a
non-$\emptyset$ character $b$ in $p(y)$. Then the extended run of
$a$'s must not only overlap $b$, but also the $\emptyset$-character
preceding $b$. The total cost of extending $a$ to overlap $b$ is
therefore $d(a, \emptyset) + d(a, b)$, which by the triangle
inequality is at least $d(\emptyset, b)$. Thus instead of
extending the run containing $a$, one could have instead extend a run
of $\emptyset$-characters to overlap $b$ at the same cost.

The fact that optimal correspondences arise by simply
extending runs of $\emptyset$-characters can then be used to prove
Theorem \ref{thmreduction00}; in particular, given such a
correspondence, one can obtain a sequence of edits from $x$ to $y$ by
performing a substitution every time the correspondence matches two
non-$\emptyset$ characters and a insertion or deletion every time the
correspondence matches a non-$\emptyset$ character and a
$\emptyset$-character.
\end{proof}

Theorem \ref{thmreduction10} proves an analogous reduction from edit
distance to LCS.

\begin{thm}
  Let $\Sigma$ be a generalized Hamming metric. Then for any $x, y \in
  \Sigma^n$, $\ed_S(p(x), p(y)) = 2 \ed(x, y)$.
  \label{thmreduction10}
\end{thm}
\begin{proof}[Proof sketch]
  Each edit in $x$ can be simulated in $s(x)$ using exactly two
  insertions/deletions. In particular, the substitution of a character
  in $x$ corresponds with the deletion and insertion of the same
  character in $s(x)$; and the insertion/deletion of a character in
  $x$ corresponds with the insertion/deletion of that character and an
  additional $\emptyset$-character in $s(x)$.

  This establishes that $\ed_S(p(x), p(y)) \le 2 \ed(x, y)$. The other
  direction of inequality is somewhat more subtle, and is differed to
  Appendix \ref{secappendixreductions}.
\end{proof}

Whereas Theorem \ref{thmreduction10} embeds edit distance into simple
edit distance with no distortion, Theorem \ref{thmreduction20} shows
that no nontrivial embedding in the other direction exists.

\begin{thm}
Consider edit distance over generalized Hamming space.  Any embedding
from edit distance to simple edit distance must have distortion at
least $2$.
\label{thmreduction20}
\end{thm}

The proofs of Theorems \ref{thmreduction00}, \ref{thmreduction10},
\ref{thmreduction20} are differed to Appendix
\ref{secappendixreductions}.

An important consequence of Theorem \ref{thmreduction00} is a new
conditional lower bound for computing DTW over a three-letter alphabet
(in which character distances are zero or one). This concludes a
direction of work initiated by Abboud, Backurs, and Williams
\cite{DTWhard2}, who proved the same result over five-letter alphabet.

\begin{cor}
 Let $\Sigma = \{a, b, c\}$ with distance function $d(a, b) = d(a, c)
 = d(b, c) = 1$. If we assume the Strong Exponential Time Hypothesis,
 then for all $\epsilon > 1$, no algorithm can compute $\dtw(x, y)$
 for $x, y \in \Sigma^n$ in time less than $O(n^{2 - \epsilon})$.
\end{cor}
\begin{proof}
  See Appendix \ref{secappendixreductions}.
\end{proof}

\section{Approximating Edit Distance Over an Arbitrary Metric}\label{seceditapprox}

In this section we present an approximation algorithm for edit
distance over an arbitrary metric space. Our algorithm achieves
approximation ratio at most $n^\epsilon$ (with high probability) and
runtime $\tilde{O}(n^{2 - \epsilon})$. Note that when the metric is a
well-separated tree metric, such an algorithm can be obtained by
combining the approximation algorithm for DTW from Section
\ref{secdtwapprox} with the reduction in Section
\ref{secreduction}. Indeed the algorithm in this section is
structurally quite similar to the one in Section \ref{secdtwapprox},
but uses a probability argument exploiting properties of edit distance
in order to hold over an arbitrary metric.

\begin{thm}
Let $(\Sigma \cup \{\emptyset\}, d)$ be an arbitrary metric space such
that $|l| \ge 1$ for all $l \in \Sigma$. For all $0 < \epsilon < 1$,
and for strings $x, y \in \Sigma^n$, there is an algorithm which
computes an $O(n^{\epsilon})$-approximation for $\ed(x, y)$ (with high
probability) in time $\tilde{O}(n^{2 - \epsilon})$.
  \label{thmeditdistanceapprox0}
\end{thm}

Using the standard dynamic-programming algorithm for computing $\ed(x,
y)$ \cite{WagnerF74, needleman1970general, vintzyuk1968}, one can
easily obtain the following observation, analogous to Lemma
\ref{lemdiagonalalgdtw0} in Section \ref{secdtwapprox}:
\begin{observation}
 Consider $x, y \in \Sigma^n$, and let $R$ be the smallest magnitude
 of the letters in $x$ and $y$. There is an $O(n^{2 - \epsilon})$-time
 algorithm which returns a value at least as large as $\ed(x, y)$; and
 which returns exactly $\ed(x, y)$ when $\ed(x, y) \le R \cdot n^{1 - \epsilon}$.
 \label{lemkeyeditobservation0}
\end{observation}

In order to prove Theorem \ref{thmeditdistanceapprox0}, we present a
new definition of the $r$-simplification of a string. The difference
between this definition and the one in the preceding section allows
the new definition to be useful when studying edit distance rather
than dynamic time warping.

\begin{defn}
  For a string $x \in \Sigma^n$ and for $r \ge 1$, we construct the
  \emph{$r$-simplification} $s_r(x)$ by removing any letter $l$ satisfying
  $|l| \le r$.
\end{defn}

In the proof of Theorem \ref{thmeditdistanceapprox0} we use
randomization in the selection of $r$ in order to ensure that $s_r(x)$
satisfies desirable properties in expectation. The key proposition
 follows:

\begin{prop}
  Consider strings $x$ and $y$ in $\Sigma^n$. Consider $0 < \epsilon <
  1$ and $R \ge 1$. Select $r$ to be a random real between $R$ and
  $2R$. Then the following three properties hold:
  \begin{itemize}
  \item Every letter $l$ in $s_r(x)$ or $s_r(y)$ satisfies $|l| \ge R$.
  \item If $\ed(x, y) \le \frac{nR}{15n^{\epsilon}}$ then $\E[\ed(s_r(x), s_r(y))] \le \frac{nR}{3n^{\epsilon}}$.
  \item If $\ed(x, y) > 5nR$, then $\ed(s_r(x), s_r(y)) > nR$.
  \end{itemize}
  \label{propthreeprops0}
\end{prop}

The full proofs of Proposition \ref{propthreeprops0} and of Theorem
\ref{thmeditdistanceapprox0} appear in Appendix
\ref{seceditapproxfull}. Structurally, both proofs are similar to the
analogous results in Section \ref{secdtwapprox}. The key difference
appears in the proof of the second part of Proposition
\ref{propthreeprops0}, which uses the random selection of $r$ in order
to probabilistically upper-bound $\ed(s_r(x), s_r(y)$. This is
presented below.

\begin{lem}
    Consider strings $x$ and $y$ in $\Sigma^n$. Consider $R \ge
    1$ and select $r$ to be a random real between $R$ and $2R$. Then
  $\E[\ed(s_r(x), s_r(y))] \le 5 \ed(x, y)$.
\end{lem}
\begin{proof}
  Consider an optimal sequence $S$ of edits from $x$ to $y$. We will
  consider the cost of simulating this sequence of edits to transform
  $s_r(x)$ to $s_r(y)$. Insertions and deletions are easily simulated
  by either performing the same operation to $s_r(x)$ or performing no
  operation at all (if the operation involves a letter of magnitude
  less than or equal to $r$). Substitutions are slightly more
  complicated as they may originally be between letters $l_1 \in x$
  and $l_2 \in y$ of different magnitudes. By symmetry, we may assume
  without loss of generality that $|l_1| < |l_2|$. We will show that
  the expected cost of simulating the substitution of $l_1$ to $l_2$
  in $s_r(x)$ is at most $5d(l_1, l_2)$. Because insertions and
  deletions can be simulated with no overhead, it follows that
  $\E[\ed(s_r(x), s_r(y))] \le 5\ed(x, y)$.

  If $|l_1| \le r < |l_2|$ then $l_1$ does not appear in $s_r(x)$
  but $l_2$ remains in $s_r(y)$. Thus what was previously a
  substitution of $l_1$ with $l_2$ becomes an insertion of $l_2$ at
  cost $|l_2|$. On the other hand, if we do not have $|l_1| \le r <
  |l_2|$, then either both $l_1$ and $l_2$ are removed from $s_r(x)$
  and $s_r(y)$ respectively, in which the substitution operation no
  longer needs to be performed, or both $l_1$ and $l_2$ are still
  present, in which case the substitution operation can still be
  performed at cost $d(l_1, l_2)$. Therefore, the expected cost of
  simulating the substitution of $l_1$ to $l_2$ in $s_r(x)$ is at most
  \begin{equation}
    \begin{split}
          \Pr[|l_1| \le r < |l_2|] \cdot |l_2| + d(l_1, l_2).
    \end{split}
    \label{eqsimcost0}
  \end{equation}

  Because $r$ is selected at random from the range $[R, 2R]$, the
  probability that $|l_1| \le r < |l_2|$ is at most $\frac{|l_2| -
    |l_1|}{R}$. By the triangle inequality, this is at most
  $\frac{d(l_1, l_2)}{R}.$ If we suppose that $|l_2| \le 4R$, then it
  follows by \eqref{eqsimcost0} that the expected cost of simulating
  the substitution of $l_1$ to $l_2$ in $s_r(x)$ is at most
 $\frac{d(l_1, l_2)}{R} \cdot 4R + d(l_1, l_2) \le 5d(l_1, l_2).$

  If, on the other hand, $|l_2| > 4R$, then in order for $|l_1| \le r$ to be true, we must have $|l_1| \le 2R$, meaning by the
  triangle inequality that $d(l_1, l_2) \ge |l_2| / 2$. Thus in this
  case $|l_2| \le 2d(l_1, l_2)$, meaning by \eqref{eqsimcost0} that the
  expected cost of simulating the substitution of $l_1$ to $l_2$ in
  $s_r(x)$ is at most three times as expensive as the original
  substitution.
\end{proof}

\paragraph*{Acknowledgments} I would like to thank Moses Charikar for his mentoring and advice throughout the project, Ofir Geri for his support and for many useful conversations, and Virginia Williams for suggesting the problem of reducing between edit distance and LCS. I would also like to thank Pawe{\l} Gawrychowski for making me aware of \cite{reduction_past}. This research is supported by an MIT Akamai Fellowship and a Fannie \& John Hertz Foundation Fellowship. This research was also supported in part by NSF Grants 1314547 and 1533644. Parts of this research were performed during the Stanford CURIS research program.

\newpage
\bibliographystyle{plain}
\bibliography{thesis}

\appendix

\section{Computing DTW in the Low-Distance Regime}\label{seclowdistancefull}
In this section, we present a low-distance regime algorithm for DTW
(with characters from an arbitrary metric in which all non-zero
distances are at least one). Given that $\dtw(x, y)$ is bounded above
by a parameter $K$, our algorithm can compute $\dtw(x, y)$ in time
$O(nK)$. Moreover, if $\dtw(x, y) > K$, then the algorithm will
conclude as much. Consequently, by doubling our guess for $K$
repeatedly, one can compute $\dtw(x, y)$ in time $O(n \cdot \dtw(x,
y))$.

Consider $x$ and $y$ of length $n$ with characters taken from a metric
space $\Sigma$ in which all non-zero distances are at least one. In
the textbook dynamic program for DTW \cite{algorithmdesign}, each pair
of indices $i, j \in [n]$ represents a subproblem $T(i, j)$ whose
value is $\dtw(x[1:i], y[1: j])$. Since $T(i, j)$ can be determined
using $T(i - 1, j), T(i, j - 1), T(i - 1, j - 1)$, and knowledge of
$x_i$ and $y_j$, this leads to an $O(n^2)$ algorithm for DTW. A common
heuristic in practice is to construct only a small band around the
main diagonal of the dynamic programming grid; by computing only
entries $T(i, j)$ with $|i - j| \le 2K$, and treating other
subproblems as having infinite return values, one can obtain a correct
computation for DTW as long as there is an optimal correspondence
which matches only letters which are within $K$ of each other in
position. This heuristic is known as the Sakoe-Chiba Band
\cite{dtwband} and is employed, for example, in the commonly used
library of Giorgino \cite{giorgino}.

The Sakoe-Chiba Band heuristic can perform badly even when $\dtw(x,
y)$ is very small, however. Consider $x = abbb\cdots b$ and $y = aaa
\cdots ab$. Although $\dtw(x, y) = 0$, if we restrict ourselves to
matching letters within $K$ positions of each other for some small
$K$, then the resulting correspondence will cost $\Omega(n)$.

In order to obtain an algorithm which performs well in the
low-distance regime, we introduce a new dynamic program for DTW. The
new dynamic program treats $x$ and $y$ asymmetrically within each
subproblem. Loosely speaking, for indices $i$ and $j$, there are
two subproblems $\tipToMiddle(x, y, i, j)$ and $\tipToMiddle(y, x, i,
j)$. The first of these subproblems evaluates to the DTW between the
first $i$ runs of $x$ and the first $j$ letters of $y$, with the added
condition that the final run of $y[1: j]$ is not extended. The second
of the subproblems is analogously defined as the DTW between the
first $i$ runs of $y$ and the first $j$ letters of $x$ with the added
condition that the final run of $x[1: j]$ is not extended.

The recursion connecting the new subproblems is somewhat more
intricate than for the textbook dynamic program. By matching the
$i$-th run with the $j$-th letter, however, we limit the number of
subproblems which can evaluate to less than $K$. In particular, if the
$j$-th letter of $y$ is in $y$'s $t$-th run, then any correspondence
which matches the $i$-th run of $x$ to the $j$-th letter of $y$ must
cost at least $\Omega(|i - t|)$ (Lemma
\ref{lemlowerboundsubproblem}). Thus for a given $j$, there are only
$O(K)$ options for $i$ such that $\tipToMiddle(x, y, i, j)$ can
possibly be at most $K$, and similarly for $\tipToMiddle(y, x, i,
j)$. Since we are interested in the case of $\dtw(x, y) \le K$, we can
restrict ourselves to the $O(nK)$ subproblems which have the potential
to evaluate to at most $O(K)$. Notice that, in fact, our algorithm
will work even when $\dtw(x, y) > K$ as long as there is an optimal
correspondence between $x$ and $y$ which only matches letters from $x$
from the $\xrun$-th run with letters from $y$ from the $\yrun$-th run
if $|\xrun - \yrun| \le O(K)$.

We will now formally present our algorithm. For the rest of the
section let $x$ and $y$ be strings of length at most $n$ and let $K$
be a parameter which we assume is greater than $\dtw(x, y)$. We will
present an $O(nK)$ time algorithm for computing $\dtw(x, y)$.

Our subproblems will be of the form $\tipToMiddle(x, y, \xrun, \yrun,
\yrunoffset)$, which is defined as follows. Let $x'$ consist of the
first $\xrun$ runs of $x$ and $y'$ consist of the first $\yrun$ runs
of $y$ until the $\yrunoffset$-th letter in the $\yrun$-th run. Then
$\tipToMiddle(x, y, \xrun, \yrun, \yrunoffset)$ is the value of the
optimal correspondence between $x'$ and $y'$ such that the $\yrun$-th
run in $y'$ is not extended.\footnote{If $\yrunoffset = 0$, then the
  $\yrun$-th run in $y'$ is empty and thus trivially cannot be
  extended.} If no such correspondence exists (which can only happen
if $\yrun \le 1$ or $\xrun = 0$), then the value of the subproblem is $\infty$. Note
that we allow $\xrun, \yrun, \yrunoffset$ to be zero, and if $\yrun$
is zero, then $\yrunoffset$ must be zero as well.

We will also consider the symmetrically defined subproblems of the
form $\tipToMiddle(y, x, \yrun, \xrun, \xoffset)$, though for
simplicity we will only prove recursions for subproblems of the first
type. (These recursions will involve subproblems of the second type,
however.)

\begin{ex}
  Suppose characters are taken from generalized Hamming space, with
  distances of zero and one. The subproblem $\tipToMiddle(efabbccccd,
  ffaabcccddd, 5, 4, 2)$ takes the value of the optimal correspondence
  between $efabbcccc$ and $ffaabcc$ such that the final $cc$ run in
  the latter is not extended. The subproblem's value turns out to be
  three, due to the following correspondence:
  \begin{center}
  \begin{tabular}{l l l l l l l l l l l}
    e&f&a&a&b&b&c&c&c&c&  \\
    f&f&a&a&b&b&b&b&c&c&.  \\
  \end{tabular}
  \end{center}
\end{ex}

Before continuing, we show how to solve for $\dtw(x, y)$ in terms of
the subproblems described above.
\begin{lem}
  Let $x$ and $y$ be strings. Then there is an optimal correspondence
  between $x$ and $y$ such that either the final run of $x$ is not
  extended, or the final run of $y$ is not extended. Consequently, if
  $x$ has $s$ runs with the final run of length $j$ and $y$ has $t$
  runs with the final run of length $k$, then
  $$\dtw(x, y) = \min\left(\tipToMiddle(x, y, s, t, k), \tipToMiddle(y, x, t, s, j)\right).$$
  \label{lemreducedcorrespondence}
\end{lem}
\begin{proof}
  Consider an optimal correspondence between $x$ and $y$. If both the final
  run of $x$ and the final run of $y$ are extended in the correspondence,
  then we can simply un-extend each by one in order to obtain a
  shorter correspondence which is at least as good. Thus any minimum-length
  optimal correspondence between $x$ and $y$ must leave either the final
  run of $x$ unextended or the final run of $y$ unextended.
\end{proof}

Because we know that $\dtw(x, y) \le K$, there is a large class of
subproblems we can ignore. This follows from the next lemma, which
provides a lower bound for $\tipToMiddle(x, y, \xrun, \yrun,
\yrunoffset)$.
\begin{lem}
  $$\tipToMiddle(x, y, \xrun, \yrun, \yrunoffset) \ge \frac{|\xrun - \yrun| - 1}{2}.$$
  \label{lemlowerboundsubproblem}
\end{lem}
\begin{proof}
Consider two strings $a$ and $b$, such that $a$ has $s$ non-empty runs
and $b$ has $t$ non-empty runs, with $s < t$. We will show that
$\dtw(a, b) \ge \frac{t - s}{2}$. Since $\tipToMiddle(x, y, \xrun,
\yrun, \yrunoffset)$ returns a correspondence between the first
$\xrun$ runs of $x$ and a portion of $y$ containing either $\yrun$ or
$\yrun - 1$ (if $\yrunoffset = 0$) non-empty runs, it follows that
$\tipToMiddle(x, y, \xrun, \yrun, \yrunoffset) \ge \frac{|\xrun -
  \yrun| - 1}{2}$, as desired.

Consider a correspondence $A$ between $a$ and $b$. Notice that $b$
contains $t - 1$ pairs of adjacent letters where the first letter is
from a different run than the second (i.e., the two letters are
distinct). The correspondence $A$ can match at most $s - 1$ of these
pairs with pairs of distinct adjacent letters from $a$. Therefore,
there are at least $t - s$ pairs of distinct adjacent letters in $b$
which are matched by $A$ with a pair of equal letters in $a$. This
means that there are at least $t - s$ pairs of adjacent letters in $b$
such that one of the two letters is mismatched by $A$. Since a
mismatched letter can appear in at most two such pairs, there must be
at least $\frac{t - s}{2}$ letters mismatched by $A$.
\end{proof}

Lemma \ref{lemlowerboundsubproblem} restricts the number of
subproblems that we are interested in to $O(nK)$. In particular,
because we assume $\dtw(x, y) \le K$, any dynamic program for
computing $\dtw(x, y)$ in terms of subproblems of the form
$\tipToMiddle(x, y, \xrun, \yrun, \yrunoffset)$ can ignore subproblems
in which $|\xrun - \yrun| > 2K + 1$ (as long as the recursion makes
each subproblem's value at least as large as any of the subproblems on
which it relies). Now consider how many subproblems $\tipToMiddle(x,
y, \xrun, \yrun, \yrunoffset)$ satisfy $|\xrun - \yrun| > 2K + 1$. The
options for $\yrun$ and $\yrunoffset$ together correspond with the
$O(n)$ positions in $y$. For each value of $\yrun$, there are $O(K)$
values of $\xrun$ such that $|\xrun - \yrun| \le 2K + 1$. Therefore,
there are $O(nK)$ subproblems of the form $\tipToMiddle(x, y, \xrun,
\yrun, \yrunoffset)$ with $|\xrun - \yrun| \le 2K + 1$. Similarly,
there are are $O(nK)$ subproblems of the form $\tipToMiddle(x, y,
\xrun, \yrun, \yrunoffset)$ with $|\xrun - \yrun| \le 2K + 1$.

In order to complete our dynamic program, it remains to present a tree
of recursions connecting the subproblems. We begin with a recursion
for the case where $\xrun, \yrun, \yrunoffset \ge 1$.

\begin{lem}
  Suppose that $\xrun$ and $\yrun$ are both between $1$ and the number
  of runs in $x$ and $y$ respectively; and that $\yrunoffset$ is
  between $1$ and the length of the $\yrun$-th run in $y$. Let
  $\xrunlength$ be the length of the $\xrun$-th run in $x$ and
  $\yrunlength$ be the length of the $\yrun$-th run in $y$. Let
  $\rundiff$ be the distance between the letter populating the
  $\xrun$-th run in $x$ and the letter populating the $\yrun$-th run
  in $y$. Then $\tipToMiddle(x, y, \xrun, \yrun, \yrunoffset)$ is given by
  $$\begin{cases}
    \min\left(\tipToMiddle(x, y, \xrun, \yrun, \yrunoffset - 1) + \rundiff,  \tipToMiddle(x, y, \xrun - 1, \yrun, \yrunoffset - \xrunlength) + \rundiff \cdot \xrunlength\right) & \text{ if }\xrunlength \le \yrunoffset \\
    \min\left(\tipToMiddle(x, y, \xrun, \yrun, \yrunoffset - 1) + \rundiff,  \tipToMiddle(y, x, \yrun - 1, \xrun, \xrunlength - \yrunoffset) + \rundiff \cdot \yrunoffset\right) & \text{ if }\xrunlength > \yrunoffset. \\
  \end{cases}$$
    \label{lemrecursion1}
\end{lem}
\begin{proof}
  Consider a minimum-cost correspondence $A$ between the first $\xrun$ runs
  of $x$ and the portion of $y$ up until the $\yrunoffset$-th letter
  in the $\yrun$-th run, such that the $\yrun$-th run in $y$ is not
  extended.

  If the $\xrun$-th run in $A$ is extended, then the cost of $A$ must
  be $\tipToMiddle(x, y, \xrun, \yrun, \yrunoffset - 1) + \rundiff$ (and
  regardless of whether the $r_x$-th run in $A$ is extended, this is an upper bound for the cost of $A$).

  If the $\xrun$-th run in $A$ is not extended, then we consider two
  cases. In the first case, $\xrunlength \le \yrunoffset$. In this
  case, The entirety of the $\xrun$-th run of $x$ is engulfed by the
  $\yrun$-th run of $y$ in the correspondence $A$. Since the $\xrun$-th run
  is not extended, the cost of the overlap is $\xrunlength \cdot
  d$. Thus the cost of $A$ must be
  $$\tipToMiddle(x, y, \xrun - 1, \yrun, \yrunoffset - \xrunlength) +
  \rundiff \cdot \xrunlength.$$ Moreover, since $A$ is minimal, as long as
  $\xrunlength \le \yrunoffset$, the cost of $A$ is at most the above
  expression, regardless of whether the $\xrun$-th run in $x$ is
  extended in $A$.

  In the second case, $\xrunlength > \yrunoffset$. In this case, the
  first $\yoffset$ letters in the $\yrun$-th run of $y$ all overlap
  the $\xrun$-th run of $x$ in $A$. Since the $\yrun$-th run is not
  extended, the cost of the overlap is $d \cdot \yoffset$. Thus, since
  the $\xrun$-th run in $x$ is also not extended in $A$, the cost of
  $A$ must be
  $$\tipToMiddle(y, x, \yrun - 1, \xrun, \xrunlength - \yrunoffset) + \rundiff \cdot
  \yoffset.$$ Moreover, since $A$ is minimal, as long as $\xrunlength
  > \yrunoffset$, the cost of $A$ is at most the above expression,
  regardless of whether the $\xrun$-th run in $x$ is extended in $A$.
\end{proof}

The cases where at least one of $\xrun, \yrun, \yrunoffset$ is $0$ are
easily handled by the following two lemmas.

\begin{lem}
  Suppose that $\xrun > 0$ is at most the number of runs in $x$, and
  $\yrun \ge 0$ is at most the number of runs in $y$. Further suppose
  that $\yrunoffset = 0$.  If $\yrun = 0$, then $\tipToMiddle(x, y,
  \xrun, \yrun, \yrunoffset) = \infty$. Otherwise, let
  $\yprevrunlength$ be the length of the $(\yrun - 1)$-th run in $y$
  (and zero if no such run exists). Then
  $$\tipToMiddle(x, y, \xrun, \yrun, \yrunoffset) = \min\left(\tipToMiddle(x, y, \xrun, \yrun - 1, \yprevrunlength), \tipToMiddle(y, x, \yrun - 1, \xrun, \xrunlength) \right).$$
    \label{lemrecursion2}
\end{lem}
\begin{proof}
  The case where $\yrun = 0$ is immediate from the fact that $\xrun >
  0$ but $\yrun = 0$.

  Suppose that $\yrun \neq 0$. Then because $\yrunoffset = 0$,
  $\tipToMiddle(x, y, \xrun, \yrun, \yrunoffset)$ is just the dynamic
  time warping distance between the first $\xrun$ runs of $x$ and the
  first $\yrun - 1$ runs of $y$. The recursion therefore follows from
  Lemma \ref{lemreducedcorrespondence}.
\end{proof}

\begin{lem}
  Suppose that $\xrun = 0$ and $\yrun \ge 0$. If $\yrun = 0$ and
  $\yoffset = 0$, then $\tipToMiddle(x, y, \xrun, \yrun, \yrunoffset)
  = 0$. Otherwise, $\tipToMiddle(x, y, \xrun, \yrun, \yrunoffset) =
  \infty$.
  \label{lemrecursion3}
\end{lem}
\begin{proof}
  This is immediate from the definition of $\tipToMiddle(x, y, \xrun,
  \yrun, \yrunoffset)$.
\end{proof}

The recursions presented in Lemmas \ref{lemrecursion1},
\ref{lemrecursion2}, \ref{lemrecursion3} cover every case where
$\xrun$ and $\yrun$ are between $0$ and the number of runs in $x$ and
$y$ respectively, and where $\yoffset$ is between $0$ and the length
of the $\yrun$-th run (or zero if $\yrun = 0$). Moreover, any
recursive call whose inputs satisfy the above conditions will only
rely on recursive subcalls whose inputs still satisfy the above
conditions. 

The recursions can be seen to be acyclic as follows. For a given
subproblem either of the form $\tipToMiddle(x, y, \xrun, \yrun,
\yrunoffset)$ or $\tipToMiddle(y, x, \yrun, \xrun, \xrunoffset)$,
associate with it a tuple $(\xrun, \yrun, t_1, t_2)$, where $t_1$ is
the length of the portion of $x$ considered by the subproblem and
$t_2$ is the length of the portion of $y$ considered by the
subproblem. The recursions given allow a subproblem to rely on another
subproblem only if the latter subproblem's tuple is dominated by the
first subproblem's tuple. This induces a partial ordering on the
subproblems, thereby forcing the recursion to be acyclic. Moreover,
because we never recurse on a case where any of $\xrun, \yrun, t_1,
t_2$ are negative, the recursion must terminate.

Algorithm \ref{algsubproblem} is a recursive algorithm for computing
$\tipToMiddle(x, y, \xrun, \yrun, \yrunoffset)$ using Lemmas
\ref{lemrecursion1}, \ref{lemrecursion2}, \ref{lemrecursion3}. The
algorithm can be transformed into a dynamic program using
memoization. Taking this approach, we can now prove the main result of
the section:
\begin{thm}
Let $x$ and $y$ be strings of length $n$ taken from a metric space
$\Sigma$ with minimum non-zero distance at least one, and let $K$ be
parameter such that $\dtw(x, y) \le K$. Then there exists a dynamic
program for computing $\dtw(x, y)$ in time $O(nK)$. Moreover, if
$\dtw(x, y) > K$, then the dynamic program will return a value greater
than $K$.
\label{thmdp}
\end{thm}
\begin{proof}
  By Lemma \ref{lemreducedcorrespondence}, if
  $x$ has $s$ runs with the final run of length $j$ and $y$ has $t$
  runs with the final run of length $k$, then
  $$\dtw(x, y) = \min\left(\tipToMiddle(x, y, s, t, k),
  \tipToMiddle(y, x, t, s, j)\right).$$

  By Lemma \ref{lemlowerboundsubproblem}, the only subproblems to have
  return value no greater than $K$ are those for which $|\xrun -
  \yrun| \le 2K + 1$. Therefore, if we relax our recursion to return
  $\infty$ whenever $|\xrun - \yrun| > 2K + 1$, we will still obtain
  the correct value for $\dtw(x, y)$ when $\dtw(x, y) \le K$, and we
  will obtain a value greater than $K$ when $\dtw(x,y) > K$. The
  resulting recursive algorithm is given by Algorithm
  \ref{algsubproblem}, and is based on the recursions given in Lemmas
  \ref{lemrecursion1}, \ref{lemrecursion2}, \ref{lemrecursion3}.

  To run the algorithm efficiently, we can use memoization in order to
  treat it as a dynamic program. Because we only recurse on cases
  where $|\xrun - \yrun| \le 2K + 1$, the total time in the spent is
  $O(nK)$. In particular, for subproblems of the form $\tipToMiddle(x,
  y, \xrun, \yrun, \yoffset)$, there are $O(n)$ options for $\yrun$
  and $\yoffset$, for each of which there are $O(K)$ options for
  $\xrun$; similarly, there are $O(nK)$ options for subproblems of the
  form $\tipToMiddle(y, x, \yrun, \xrun, \xoffset)$.
\end{proof}

\begin{cor}
Let $x$ and $y$ be strings of length $n$. Then one can compute
$\dtw(x, y)$ in time $O(n \cdot \dtw(x, y))$.
\end{cor}
\begin{proof}
  For a parameter $K$, the dynamic program from the proof of Theorem
  \ref{thmdp} correctly computes $\dtw(x, y)$ when $\dtw(x, y) \le K$,
  and otherwise returns a value greater than $K$.
  
  By guessing $K$ to be successive powers of two, one can find the
  first power of two $K$ such that $\dtw(x, y) \le K$ in time $O(n
  \cdot 1 + n \cdot 2 + n \cdot 4 + \cdots + n \cdot K) = O(n \cdot
  K)$, allowing one to then compute $\dtw(x, y)$ in time $O(n \cdot
  K) = O(n \cdot \dtw(x, y))$.
\end{proof}

\begin{algorithm}[h]
\caption{Subproblem $\tipToMiddle(x, y, \xrun, \yrun, \yoffset)$ with
  the relaxation $\tipToMiddle(x, y, \xrun, \yrun, \yoffset) = \infty$
  when $|\xrun - \yrun| > 2K + 1$.}\label{algsubproblem}
\begin{algorithmic}[1]
\Procedure{$\tipToMiddle$}{$x$, $y$, $\xrun$, $\yrun$, $\yoffset$}
\If {$|\xrun - \yrun| > 2K + 1$}
\State \Return $\infty$.
\EndIf
\If{$\xrun > 0$ and $\yoffset > 0$}
\State $\rundiff \gets \text{ distance between each letter in }\xrun\text{-th run in }x\text{ and each letter in }\yrun\text{-th run in }y$.
\State $\xrunlength \gets \text{ length of }\xrun\text{-th run in }x.$
\State $\yrunlength \gets \text{ length of }\yrun\text{-th run in }y.$
\State $\text{answer} \gets \tipToMiddle(x, y, \xrun, \yrun, \yrunoffset - 1) + \rundiff$.
\If{$\xrunlength \le \yrunoffset$}
\State $\text{option2} \gets \tipToMiddle(x, y, \xrun - 1, \yrun, \yrunoffset - \xrunlength) + \rundiff \cdot \xrunlength$.
\State $\text{answer} \gets \min\left(\text{answer}, \text{option2}\right)$
\EndIf
\If{$\xrunlength > \yrunoffset$}
\State $\text{option2} \gets \tipToMiddle(y, x, \yrun - 1, \xrun, \xrunlength - \yrunoffset) + \rundiff \cdot \yrunoffset$.
\State $\text{answer} \gets \min\left(\text{answer}, \text{option2}\right)$
\EndIf
\State \Return $\text{answer}$.
\EndIf
\If{$\xrun > 0$ and $\yrun = 0$}
\State \Return $\infty$.
\EndIf
\If{$\xrun > 0$ and $\yrun > 0$ and $\yoffset = 0$}
\State $\xrunlength \gets \text{ length of }\xrun\text{-th run in }x.$
\State $\yprevrunlength \gets \text{ length of }(\yrun - 1)\text{-th run in }y.$
\State $\text{option1} \gets \tipToMiddle(x, y, \xrun, \yrun - 1, \yprevrunlength)$.
\State $\text{option2} \gets \tipToMiddle(y, x, \yrun - 1, \xrun, \xrunlength)$.
\State \Return $\min(\text{option1}, \text{option2})$.
\EndIf
\If{$\xrun = 0$ and $\yrun = 0$}
\State \Return $0$.
\EndIf
\State \Return $\infty$.
\EndProcedure
\end{algorithmic}
\end{algorithm}

\newpage

\section{Approximating DTW Over Well-Separated Tree Metrics}\label{secdtwapproxfull}

In this section, we present an $\tilde{O}(n^{2 - \epsilon})$-time
$O(n^{\epsilon})$-approximation algorithm for DTW over a
well-separated tree metric with logarithmic depth. We begin by
presenting a brief background on well-separated tree metrics.

\begin{defn}
  Consider a tree $T$ whose vertices form an alphabet $\Sigma$, and
  whose edges have positive weights. $T$ is said to be a
  \emph{well-separated tree metric} if every root-to-leaf path
  consists of edges ordered by nonincreasing weight. The
  \emph{distance} between two nodes $u, v \in \Sigma$ is defined as
  the maximum weight of any edge in the shortest path from $u$ to $v$.

\end{defn}

Well-separated tree metrics are universal in the sense that any metric
$\Sigma$ can be efficiently embedded (in time $O(|\Sigma|^2)$) into a
well-separated tree metric $T$ with expected distortion $O(\log
|\Sigma|)$ \cite{trees}.\footnote{In particular, \cite{trees} shows
  how to embed $\Sigma$ into what is known as a \emph{2-hierarchically
    well-separated tree metric}. These metrics differ slightly from
  our definition of a well-separated tree metric in that the distance
  between two nodes is given by the \emph{sum} of the edge weights in
  the shortest path, rather than the \emph{maximum}. It turns out that
  in the case of $2$-hierarchically well-separated tree metrics, these
  quantities will differ from one another by at most a factor of two,
  however.} That is, if $d_\Sigma(u, v)$ denotes the distance between
two elements of $\Sigma$, and $d_T(u, v)$ denotes the distance in $T$
between the nodes corresponding with $u$ and $v$, then
\begin{equation}
  d_\Sigma(u, v) \le d_T(u, v),
  \label{eq1}
\end{equation}
and
\begin{equation}
  \E[d_T(u, v)] \le O(\log n) \cdot d_\Sigma(u, v).
  \label{eq2}
\end{equation}
Moreover, using the algorithm from Theorem 8 of \cite{bansal2011} to
modify the well-separated tree metric $T$, we may assume without loss
of generality that $T$ has maximum depth $O(\log |\Sigma|)$, allowing
for fast evaluation of distances. 

For strings $x, y \in \Sigma^n$, let $\dtw_T(x, y)$ denote the dynamic
time warping distance after embedding $\Sigma$ into $T$. Then by
\eqref{eq1}, every correspondence between the embedded strings must
cost at least as much as the same correspondence between the original
strings, meaning that
\begin{equation}
  \dtw(x, y) \le \dtw_T(x, y).
  \label{eq3}
\end{equation}
Moreover, by \eqref{eq2}, if we take the optimal correspondence
between $x$ and $y$, then the same correspondence between the embedded
strings will have expected cost at most $O(\log n)$ times as large,
meaning that
\begin{equation}
  \E[\dtw_T(x, y)] \le \dtw(x, y).
  \label{eq4}
\end{equation}

Combined, \eqref{eq3} and \eqref{eq4} tell us that any approximation
algorithm for DTW over well-separated tree metrics will immediately
yield an approximation algorithm over an arbitrary polynomial-size
metric $\Sigma$, with two caveats: the new algorithm will have its
multiplicative error increased by $O(\log n)$; and $O(\log n)$
instances of $\Sigma$ embedded into a well-separated tree metric must
be precomputed for use by the algorithm (requiring, in general,
$O(|\Sigma|^2 \log n)$ preprocessing time). In particular, given $O(\log
n)$ tree embeddings of $\Sigma$, $T_1, \ldots, T_{O(\log n)}$, \eqref{eq3}
and \eqref{eq4} tell us that with high probability $\min_i
\left(\dtw_{T_i}(x, y)\right)$ will be within a logarithmic factor of
$\dtw(x, y)$.

The remainder of the section will be devoted to designing an
approximation algorithm for DTW over a well-separated tree
metric. We will prove the following theorem:

\begin{thm}
  Consider $0 < \epsilon < 1$. Suppose that $\Sigma$ is a
  well-separated tree metric of polynomial size and at most
  logarithmic depth. Moreover, suppose that the aspect ratio of
  $\Sigma$ is at most exponential in $n$ (i.e., the ratio between the
  largest distance and the smallest non-zero distance). Then in time
  $\tilde{O}(n^{2 - \epsilon})$ we can obtain an
  $O(n^{\epsilon})$-approximation for $\dtw(x, y)$ for any $x, y \in
  \Sigma^n$.
  \label{thmdtwapproxtree}
\end{thm}

An important consequence of the theorem occurs for DTW over the reals.

\begin{cor}
  Consider $0 < \epsilon < 1$. Suppose that $\Sigma$ is an
  $O(n)$-point subset of the reals with polynomial aspect ratio. Then
  in time $\tilde{O}(n^{2 - \epsilon})$ we can obtain an
  $O(n^{\epsilon})$-approximation for $\dtw(x, y)$ with high
  probability for any $x, y \in \Sigma^n$.
\end{cor}
\begin{proof}
  Without loss of generality, $\Sigma \subset [0, n^c]$ for some
  constant $c$, and no two points in $\Sigma$ have distance less than
  one from each other.
  
 The key observation is that there is an $O(n\log n)$-time embedding
 with $O(\log n)$ expected distortion from $\Sigma$ to a
 well-separated tree metric of size $O(n)$ with logarithmic
 depth. Such an embedding can be constructed by a random partitioning
 process, as follows\footnote{This construction appears to be somewhat
   of a folklore result. It is a natural consequence of ideas from
   \cite{DefnExpectedDistortion}.}. If $m_1$ and $m_2$ are the
 smallest and largest elements of $\Sigma$, respectively, and $r = m_2
 - m_1$ is the size of the range in which $\Sigma$ lies, then first
 select a pivot $p$ at random from the range $[m_1 + r / 4, m_2 - r /
   4]$. Next recursively constructing a subtree $L$ for the subset of
 $\Sigma$ less than $p$, and then recursively constructing a subtree
 $R$ for the subset of $\Sigma$ greater than $p$. Finally, return the
 tree $T$ in which $p$ is the root, $L$ is the left subtree, $R$ is
 the right subtree, and the edges from $p$ to its children have weight
 $r$. One can easily verify that $T$ is a well-separated tree metric;
 that distances in $T$ dominate distances in $\Sigma$; that the depth
 of $T$ is $O(\log n)$; and that $T$ can be recursively constructed in
 time $O(n \log n)$. It is slightly more subtle to show that for $a, b
 \in \Sigma$, the expected distance $d_T(a, b)$ between $a$ and $b$ in
 $T$ is at most $O(\log n) \cdot |a - b|$. Note that $d_T(a, b)$ is
 the length $r$ of the range for the recursive subproblem in which $a$
 and $b$ are split from one-another. For each recursive subproblem
 containing $a$, the probability that $a$ and $b$ will be split by the
 pivot in that subproblem is at most $O(|a - b| / r)$, where $r$ is
 the range size of the subproblem. This means that the expected
 contribution to $d_T(a, b)$ by that subproblem is $O(|a - b| / r)
 \cdot r \le O(|a - b|)$. Since there are at most $O(\log n)$
 recursive subproblems containing $a$, the expected value of $d_T(a,
 b)$ is at most $O(\log n) \cdot |a - b|$. This completes the analysis
 of the embedding.

 Combining the above embedding with Theorem \ref{thmdtwapproxtree}
 (and taking the minimum output of the algorithm over $O(\log n)$
 independent iterations in order to obtain a high probability result)
 yields the desired result.
\end{proof}

In proving Theorem \ref{thmdtwapproxtree}, our approximation algorithm
will take advantage of what we refer to as the
\emph{$r$-simplification} of a string over a well-separated tree
metric.
\begin{defn}
Let $T$ be a well-separated tree metric whose nodes form an alphabet $\Sigma$. For a
string $x \in \Sigma^n$, and for any $r \ge 1$, the
\emph{$r$-simplification} $s_r(x)$ is constructed by replacing
each letter $l \in x$ with its highest ancestor $l'$ in $T$ that can
be reached from $l$ using only edges of weight at most $r / 4$.
\end{defn}

Our approximation algorithm will apply the low-distance regime
algorithm from the previous section to $s_r(x)$ and $s_r(y)$ for
various $r$ in order to extract information about $\dtw(x, y)$. Notice
that using our low-distance regime algorithm for DTW, we get the
following useful lemma for free:

\begin{lem}
  Consider $0 < \epsilon < 1$. Suppose that for all pairs $l_1, l_2$
  of distinct letters in $\Sigma$, $d(l_1, l_2) \ge \gamma$. Then for
  $x, y \in \Sigma^n$ there is an $O(n^{2 - \epsilon})$ time algorithm
  which either computes $\dtw(x, y)$ exactly, or concludes that $\dtw(x,
  y) > \gamma n^{1 - \epsilon}$.
  \label{lemdiagonalalgdtw}
\end{lem}
\begin{proof}
  This follows immediately from Theorem \ref{thmdp}.
\end{proof}

The next lemma states three important properties of
$r$-simplifications. We remark that the same lemma appears in our
concurrent work on the communication complexity of DTW, in which we
use the lemma in designing an efficient
one-way communication protocol \cite{dtwcomm}.

\begin{lem}
  Let $T$ be a well-separated tree metric with distance function $d$
  and whose nodes form the alphabet $\Sigma$. Consider strings $x$ and
  $y$ in $\Sigma^{n}$.

  Then the following three properties of $s_r(x)$ and $s_r(y)$ hold:
  \begin{itemize}
  \item For every letter $l_1 \in s_r(x)$ and every letter $l_2 \in
    s_r(y)$, if $l_1 \neq l_2$, then $d(l_1, l_2) > r / 4$.
  \item For all $\alpha$, if $\dtw(x, y) \le nr / \alpha$ then $\dtw(s_r(x), s_r(y)) \le nr / \alpha$.
  \item If $\dtw(x, y) > nr$, then $\dtw(s_r(x), s_r(y)) > nr/2$. 
    \end{itemize}
  \label{lemthreeprops}
\end{lem}
\begin{proof}
The first part of the lemma follows directly from the definitions of
$s_r(x)$ and $s_r(y)$.

Let $C$ be the optimal correspondence between $x$ and $y$, and let
$C'$ be the same correspondence between $s_r(x)$ and $s_r(y)$. Suppose
$C$ matches some letter $l_1 \in x$ to a letter $l_2$ in $y$. Let
$l_1'$ and $l_2'$ be the corresponding letters in $s_r(x)$ and
$s_r(y)$. Notice that if $d(l_1, l_2) \le r / 4$, then
\begin{equation}
  l_1' = l_2',
\label{eqequalityofl}
\end{equation}
and that if $d(l_1, l_2) > r / 4$, then
\begin{equation}
  d(l_1', l_2') = d(l_1, l_2).
  \label{eqequalityofdist}
\end{equation}

By \eqref{eqequalityofl} and \eqref{eqequalityofdist}, the
correspondence $C'$ costs no more than $C$, meaning that $\dtw(s_r(x),
s_r(y)) \le \dtw(x, y)$. Therefore, if $\dtw(x, y) \le nr / \alpha$
then $\dtw(s_r(x), s_r(y)) \le nr / \alpha$, completing the second
part of the lemma.

Now suppose that $\dtw(x, y) > nr$. Consider an optimal correspondence
$D$ between $s_r(x)$ and $s_r(y)$, and assume without loss of
generality that $D$ is of length no more than $2n$. Equations
\eqref{eqequalityofl} and \eqref{eqequalityofdist} tell us that the
cost of $D$ can be no more than $2n \cdot r / 4$ smaller than the cost
of the same correspondence between $x$ and $y$. Since $\dtw(x, y) >
nr$, it follows that $\dtw(s_r(x), s_r(y)) > nr - 2n \cdot r / 4 = nr
/ 2$, completing the third part of the lemma.
\end{proof}

We can now prove Theorem \ref{thmdtwapproxtree}.

\begin{proof}[Proof of Theorem \ref{thmdtwapproxtree}]
Without loss of generality, the minimum non-zero distance in $\Sigma$ is 1 and
the largest distance is some value $m$, which is at most exponential
in $n$.

We begin by defining the \emph{$(r, n^{\epsilon})$-DTW gap} problem
for $r \ge 1$, in which for two strings $x$ and $y$ a return value of
0 indicates that $\dtw(x, y) < nr$ and a return value of 1 indicates
that $\dtw(x, y) \ge n^{1 - \epsilon}r$. By Lemma \ref{lemthreeprops},
in order to solve the $(r, n^{\epsilon})$-DTW gap problem for $x$ and
$y$, it suffices to determine whether $\dtw(s_r(x), s_r(y)) \le n^{1 -
  \epsilon}r$. Moreover, because the minimum distance between distinct
letters in $s_r(x)$ and $s_r(y)$ is at least $r / 4$, this can be done
in time $O(n^{2 - \epsilon} \log n)$ using Lemma
\ref{lemdiagonalalgdtw}.\footnote{The logarithmic factor comes from
  the fact that evaluating distances between points may take
  logarithmic time in our well-separated tree metric.}

In order to obtain an $n^{\epsilon}$-approximation for $\dtw(x, y)$,
we begin by using Lemma \ref{lemdiagonalalgdtw} to either determine
$\dtw(x, y)$ or to determine that $\dtw(x, y) \ge n^{1 -
  \epsilon}$. For the rest of the proof, suppose we are in the latter
case, meaning that we know $\dtw(x, y) \ge n^{1 - \epsilon}$.

We will now consider the $(2^i, n^{\epsilon} / 2)$-DTW gap problem for
$i \in \{0, 1, 2, \ldots, \lceil \log m \rceil\}$. If the $(2^0,
n^{\epsilon} / 2)$-DTW gap problem returned 0, then we would know that
$\dtw(x, y) \le n$, and thus we could return $n^{1 - \epsilon}$ as an
$n^{\epsilon}$-approximation for $\dtw(x, y)$. Therefore, we need only
consider the case where the $(2^0, n^{\epsilon} / 2)$-DTW gap returns
$1$. Moreover we may assume without computing it that $(2^{ \lceil
  \log m \rceil}, n^{\epsilon}/2)$-DTW gap returns 0 since trivially
$\dtw(x, y)$ cannot exceed $nm$. Because $(2^i, n^{\epsilon} / 2)$-DTW
gap returns 1 for $i = 0$ and returns $0$ for $i = \lceil \log m
\rceil$, there must be some $i$ such that $(2^{i - 1}, n^{\epsilon} /
2)$-DTW gap returns $1$ and $(2^{i}, n^{\epsilon} / 2)$-DTW gap
returns 0. Moreover, we can find such an $i$ by performing a binary
search on $i$ in the range $R = \{0, \ldots, \lceil \log m \rceil\}$. We
begin by computing $(2^i, n^{\epsilon} / 2)$-DTW gap for $i$ in the
middle of the range $R$. If the result is a one, then we can recurse
on the second half of the range; otherwise we recurse on the first
half of the range. Continuing like this, we can find in time
$\tilde{O}(n^{2 - \epsilon} \log \log m) = \tilde{O}(n^{2 -
  \epsilon})$ some value $i$ for which $(2^{i - 1}, n^{\epsilon} /
2)$-DTW gap returns $1$ and $(2^{i}, n^{\epsilon} / 2)$-DTW gap
returns 0. Given such an $i$, we know that $\dtw(x, y) \ge \frac{2^{i
    - 1}n}{n^{\epsilon} / 2} = 2^i n^{1 - \epsilon}$ and that $\dtw(x,
y) \le 2^in$. Thus we can return $2^i n^{1 - \epsilon}$ as an
$n^{\epsilon}$ approximation of $\dtw(x, y)$.
\end{proof}
\section{Reducing Edit Distance to DTW and LCS}\label{secappendixreductions}

In this section we present a simple reduction from edit distance over
an arbitrary metric to DTW over the same metric.
At the end of the section, we
prove as a corollary a conditional lower bound for DTW over
three-letter Hamming space, prohibiting any algorithm from running in
strongly subquadratic time.

Surprisingly, the \emph{exact same reduction}, although with a
different analysis, can be used to reduce the computation of edit
distance (over generalized Hamming space) to the computation of
longest-common-subsequence length (LCS). Since computing LCS is
equivalent to computing edit distance without substitutions, this
reduction can be interpreted as proving that edit distance without
substitutions can be used to efficiently simulate edit distance with
substitutions, also known as \emph{simple edit distance}.

Recall that for a metric $\Sigma \cup \{\emptyset\}$, we define the
edit distance between two strings $x, y \in \Sigma^n$ such that the
cost of a substitution from a letter $l_1$ to $l_2$ is $d(l_1, l_2)$,
and the cost of a deletion or insertion of a letter $l$ is $d(l,
\emptyset)$. Additionally, define the \emph{simple edit distance}
$\ed_S(x, y)$ to be the edit distance using only insertions and
deletions.

For a string $x \in \Sigma^n$, define the \emph{padded string} $p(x)$
of length $2n + 1$ to be the string $\emptyset x_1 \emptyset x_2
\emptyset x_3 \cdots x_n \emptyset$. In particular, for $i \le 2n +
1$, $p(x)_i = \emptyset$ when $i$ is odd, and $p(x)_i = x_{i / 2}$
when $i$ is even. The following theorem proves that $\dtw(p(x), p(y))
= \ed(x, y)$.

\begin{thm}
  Let $\Sigma \cup \{\emptyset\}$ be a metric. Then for any $x, y \in
  \Sigma^n$,
  $\dtw(p(x), p(y)) = \ed(x, y).$
  \label{thmreduction0}
\end{thm}
\begin{proof}
  We begin by showing that there is an optimal correspondence $C$
  between $p(x)$ and $p(y)$ in which the only extended runs are those
  consisting of the letter $\emptyset$. Consider an arbitrary optimal
  correspondence $D$ between $p(x)$ and $p(y)$, and assume without
  loss of generality that no two runs which are extended in the
  correspondence are matched by the correspondence to overlap. Suppose
  a run $r$ consisting of a letter $a \neq \emptyset$ in $x$ is
  extended in $D$. If the first letter in the run $r$ is matched with
  the letter $\emptyset$ in $y$, then the correspondence $D$ could be
  improved by replacing the first letter in $r$ with an $\emptyset$
  character, which can be achieved by un-extending the run $r$ by one,
  and instead further extending the preceding run of $\emptyset$'s by
  one. Consequently, the first letter in the run $r$ must match to
  some letter $l_1 \neq \emptyset$ in $y$, and similarly the final
  letter in the run $r$ must match to some letter $l_2 \neq \emptyset$
  in $y$, meaning the extended run must be at least three letters
  long. Now consider the first two letters $l_1, \emptyset$ to which
  the run $r$ is matched by $D$. By the triangle inequality, $d(l_1,
  a) + d(a, \emptyset) \ge d(l_1, \emptyset)$ (recall that $a$ is the
  letter which populates the run $r$). It follows that if we un-extend
  the run $r$ by two letters, and instead further extend the run of
  $\emptyset$'s preceding $r$, then we arrive at a correspondence $D'$
  no more expensive than $D$. Moreover, the sum of the lengths of runs
  of non-$\emptyset$ elements in $D'$ has been reduced by two from the
  same sum for $D$. Repeating this process as many times as necessary,
  we can arrive at a correspondence $C$ in which no runs containing
  non-$\emptyset$ letters are extended, as desired. Because $D$ was an
  optimal correspondence, and $C$ costs no more than $D$, $C$ must
  also be optimal.

  Using the correspondence $C$, we can now prove that $\ed(x, y) \le
  \dtw(p(x), p(y))$. In particular, we can construct a sequence of
  edits between $x$ and $y$ at most as expensive as the correspondence
  $C$. To do this, we first delete from $x$ any letter $l$ in $x$
  which is matched by $C$ to an $\emptyset$ character, and we do the
  same for $y$. Call the resulting strings $x'$ and $y'$. Notice that
  for each letter $x_i'$ in $x'$, $C$ must match $x_i'$ to the
  corresponding letter $y_i'$ in $y'$. Consequently, if we perform
  substitutions in order to transform each $x_i'$ into $y_i'$, then
  the full sequence of edits between $x$ and $y$ will have cost no
  more than the cost of $C$.

  Finally, to show that $\dtw(p(x), p(y)) \le \ed(x, y)$, we present a
  correspondence $C$ between $p(x)$ and $p(y)$ of cost $\ed(x,
  y)$. Consider an optimal sequence of edits $E$ from $x$ to $y$. One
  can think of $E$ as performing a series of insertions in $x$, a
  series of insertions in $y$, and then a series of substitutions to
  transform the resulting strings into one another. Now suppose that
  for each insertion into $x$, we instead insert an $\emptyset$
  character, and similarly for each insertion into $y$. Notice that
  the resulting strings $x'$ and $y'$ will satisfy
  $\sum_i d(x'_i, y'_i) = \text{cost}(E).$ Next, insert an additional
  $\emptyset$ character in every other position in $x'$ and in every other
  position in $y'$ to obtain strings $x''$ and $y''$. These new
  strings still satisfy
  $\sum_i d(x''_i, y''_i) = \text{cost}(E),$ but have the additional
  property that they are expansions of $p(x)$ and $p(y)$. Hence
  $\dtw(p(x), p(y)) \le \ed(x, y)$, completing the proof.
  
\end{proof}

Theorem \ref{thmreduction1} proves an analogous reduction from edit
distance to LCS. As a convention, we use \emph{complex edits} to refer
to insertions, deletions, and substitutions, and \emph{simple edits}
to refer to edits consisting only of insertions and deletions.

\begin{thm}
  Let $\Sigma$ be a generalized Hamming metric. Then for any $x, y \in
  \Sigma^n$, $\ed_S(p(x), p(y)) = 2 \ed(x, y)$.
  \label{thmreduction1}
\end{thm}
\begin{proof}
Given a sequence of complex edits from $x$ to $y$, those edits can be
emulated at twice the cost using simple edits from $p(x)$ to
$p(y)$. In particular, the insertion of a letter $l$ becomes the
insertion of the letters $l,\emptyset$, the deletion of a letter $l$
becomes the deletion of the letters $l,\emptyset$, and the
substitution of the letter $l$ to the letter $k$ becomes the deletion
and insertion needed to transform $l,\emptyset$ to $k,\emptyset$. It
follows that $\ed_S(p(x), p(y)) \le 2\ed(x, y)$.

It remains to prove that $2\ed(x, y) \le \ed_S(p(x), p(y))$. For the
sake of completeness, we will now prove it formally. Consider an
optimal alignment $A$ between $p(x)$ and $p(y)$ (allowing only simple
edits). That is, $A$ is a non-crossing bipartite graph between the
letters of $p(x)$ and $p(y)$ with edges only between letters of the same
value. The cost of $A$ is the number of singleton nodes in $A$, which
is equal to $\ed_S(p(x), p(y))$. Without loss of generality, for each edge
$e_1$ in $A$ connecting two non-null nodes $u$ and $v$, there is
another edge $e_2$ connecting the $0$-nodes directly following $u$ and
$v$. Indeed, if there is not, then since $A$ is non-crossing, at most
one of the two $0$-nodes can be part of an edge in $A$. Deleting such
an edge and replacing it with the desired edge does not change the
cost of $A$.

Call a letter $x_i$ in $x$ \emph{totally unmatched by $A$} if
$p(x)_{2i}$ and $p(x)_{2i + 1}$ are both unmatched in $A$. Call a
letter $x_i$ \emph{partially matched to $y_j$ by $A$} if $p(x)_{2i}$
is unmatched in $A$ but $p(x)_{2i + 1}$ is matched to some $p(y)_{2j +
  1}$. Call a letter $x_i$ \emph{totally matched to $y_j$ by $A$} if
both $p(x)_{2i}$ and $p(x)_{2i + 1}$ are matched by $A$ to some
$p(y)_{2j}$ and $p(y)_{2j + 1}$ respectively. If we define the
analogous terms for letters of $y$, then notice that every letter of
$x$ and $y$ is either totally unmatched, partially matched to another
partially matched letter, or totally matched to another totally
matched letter.

We now construct an alignment between $x$ and $y$ (allowing for
complex edits). That is, we construct a non-crossing bipartite graph
$B$ from the letters of $x$ to the letters of $y$, such that the
number of unmatched nodes in $B$ plus the number of edges between
different-valued characters in $B$ is at most the number of unmatched
nodes in $A$. For each letter $x_i$ or $y_i$ that is totally unmatched
by $A$, leave it as a singleton node in $B$. For each letter $x_i$
which is partially matched to some $y_j$ by $A$, match $x_i$ to $y_j$
in $B$. For each letter $x_i$ which is totally matched to some letter
$y_j$, match $x_i$ to $y_j$ in $B$. The resulting $B$ has cost equal
to the number of totally unmatched letters in $x$ and $y$, plus the
number of partially matched letters in $x$ (since they are paired with
the partially matched letters in $y$ through substitutions). This is
half the cost of $A$, since each totally unmatched letter in $x$ or
$y$ corresponds with two adjacent singleton nodes in $A$, and each
partially matched pair of letters between $x$ and $y$ corresponds with
two singleton nodes in $A$. It follows that $2 \ed(x, y) \le
\ed_S(p(x), p(y))$.

\end{proof}

Whereas Theorem \ref{thmreduction1} embeds edit distance into simple
edit distance with no distortion, the next theorem shows that no
nontrivial embedding in the other direction exists.

Formally, we say that an \emph{embedding} from a metric space $(M_1,
d_1)$ to $(M_2, d_2)$ is an injective map $\phi$. The
\emph{distortion} of an embedding is defined as
$$\frac{\sup_{x, y \in M_1} d_1(x, y) / d_2(\phi(x),
  \phi(y))}{\inf_{x, y \in M_1} d_1(x, y) / d_2(\phi(x), \phi(y))},$$
unless the numerator is unbounded or the denominator is zero, in which
case the distortion is $\infty$.

A trivial embedding from edit distance to simple edit distance would
be the identity map, which achieves distortion $2$. Theorem
\ref{thmreduction2} establishes that no other embedding can do better.

\begin{thm}
Consider edit distance over generalized Hamming space.  Any embedding
from edit distance to simple edit distance must have distortion at
least $2$.
\label{thmreduction2}
\end{thm}
\begin{proof}
  For an alphabet $\Sigma$ containing both $0$ and $1$, suppose for
  contradiction that there is an embedding $\phi: \Sigma^* \rightarrow
  \Sigma^*$ such that $c_1 \ed_S(x, y) \le \ed(\phi(x), \phi(y)) \le
  c_2 \ed_S(x, y)$ for constants $c_1 \le c_2$ within a factor of less
  than two of each other.

  For all $n$, $\ed(\phi(0^n), \phi(1^n)) \ge c_1 \cdot 2n$. Since the
  edit distance between two strings is at most the maximum of their
  lengths, it follows without loss of generality that $|\phi(0^n)| \ge
  c_1\cdot 2 n$ (since the case where $|\phi(1^n)| \ge c_1 \cdot 2 n$
  is symmetric).

  Now define $x$ to be the empty string, and $y = 0^n$. By assumption,
  $\ed(\phi(x), \phi(y)) \le c_2n$, and thus $\phi(x)$ and $\phi(y)$
  are within $c_2n$ of each other in length. Since $\phi(y)$ is length
  at least $2 \cdot c_1 n$, it follows that $\phi(x)$ must be length
  at least $(2c_1 - c_2)n$. Recall that $c_1$ and $c_2$ are within
  less than a factor of two of each other, meaning that $(2c_1 - c_2)$
  is a positive constant. Therefore, we have shown that for all $n$,
  $|\phi(x)| \ge \Omega(n)$, a contradiction.

\end{proof}

We conclude the section by obtaining a novel conditional lower bound
for computing DTW over a three-letter alphabet (in which character
distances are zero or one). This concludes a direction of work
initiated by Abboud, Backurs, and Williams \cite{DTWhard2}, who proved
the same result over five-letter alphabet.

\begin{cor}
 Let $\Sigma = \{a, b, c\}$ with distance function $d(a, b) = d(a, c)
 = d(b, c) = 1$. If we assume the Strong Exponential Time Hypothesis,
 then for all $\epsilon > 1$, no algorithm can compute $\dtw(x, y)$
 for $x, y \in \Sigma^n$ in time less than $O(n^{2 - \epsilon})$.
\end{cor}
\begin{proof}
  By Theorem 1.2 of \cite{DTWhard}, edit distance between binary
  strings cannot be computed in strongly subquadratic time, assuming
  the Strong Exponential Time Hypothesis. Applying Theorem
  \ref{thmreduction0}, we get the same result for DTW over $\Sigma$.
\end{proof}

\section{Approximating Edit Distance Over an Arbitrary Metric}\label{seceditapproxfull}

In this section we present an approximation algorithm for edit
distance over an arbitrary metric space. Our algorithm achieves
approximation ratio at most $n^\epsilon$ (with high probability) and
runtime $\tilde{O}(n^{2 - \epsilon})$. Note that when the metric is a
well-separated tree metric, such an algorithm can be obtained by
combining the approximation algorithm for DTW from Section
\ref{secdtwapprox} with the reduction in Section
\ref{secreduction}. Indeed the algorithm in this section is
structurally quite similar to the one in Section \ref{secdtwapprox},
but uses a probability argument exploiting properties of edit distance
in order to hold over an arbitrary metric.

\begin{thm}
Let $(\Sigma \cup \{\emptyset\}, d)$ be an arbitrary metric space such
that $|l| \ge 1$ for all $l \in \Sigma$. For all $0 < \epsilon < 1$,
and for strings $x, y \in \Sigma^n$, there is an algorithm which
computes an $O(n^{\epsilon})$-approximation for $\ed(x, y)$ (with high
probability) in time $\tilde{O}(n^{2 - \epsilon})$.
  \label{thmeditdistanceapprox}
\end{thm}

Before continuing, we remind the reader of the standard
dynamic-programming algorithm for computing $\ed(x, y)$
\cite{WagnerF74, needleman1970general, vintzyuk1968}. For $i, j
\in \{0, \ldots, n\}$, the subproblem $\mathcal{A}(i, j)$ is defined
to be the edit distance
$$\ed(x_1 \cdots x_i, y_1 \cdots y_j).$$ Notice that for $i, j > 0$,
the return-value for $\mathcal{A}(i, j)$ is completely determined by
$\mathcal{A}(i - 1, j)$, $\mathcal{A}(i, j - 1)$, and $\mathcal{A}(i -
1, j - 1)$, as well as knowledge of $x_i$ and $y_j$. Using this, one
can formulate a recursion which results in an $O(n^2)$-time
dynamic-program for computing $\ed(x, y)$.

Now consider the same dynamic program, except with $\mathcal{A}(i, j)$
artificially set to $\infty$ whenever $|i - j| > K$, for some
parameter $K$. This new dynamic program runs in time $O(n \cdot K)$
and returns the minimum cost of any sequence of edits that transforms
$x$ to $y$ without ever matching any $x_i$ to some $y_j$ for which $|i
- j| > K$. Importantly, this means that if there is an optimal
sequence of edits involving no more than $K$ insertions and deletions,
then the new dynamic program will find the true value of $\ed(x, y)$;
otherwise, the dynamic program may return an overestimate for $\ed(x,
y)$. This implies in the following observation:

\begin{observation}
 Consider $x, y \in \Sigma^n$, and let $R$ be the smallest magnitude
 of the letters in $x$ and $y$. There is an $O(n^{2 - \epsilon})$-time
 algorithm which returns a value at least as large as $\ed(x, y)$; and
 which returns exactly $\ed(x, y)$ when $\ed(x, y) \le R \cdot n^{1 - \epsilon}$.
 \label{lemkeyeditobservation}
\end{observation}

In order to prove Theorem \ref{thmeditdistanceapprox}, we present a
new definition of the $r$-simplification of a string. The difference
between this definition and the one in the preceding section allows
the new definition to be useful when studying edit distance rather
than dynamic time warping.

\begin{defn}
  For a string $x \in \Sigma^n$ and for $r \ge 1$, we construct the
  \emph{$r$-simplification} $s_r(x)$ by removing any letter $l$ satisfying
  $|l| \le r$.
\end{defn}

The proof of Theorem \ref{thmeditdistanceapprox} has a similar
structure to that of Theorem \ref{thmdtwapproxtree}, but interestingly
avoids the use of well-separated tree metrics.

A key insight in the proof of Theorem \ref{thmeditdistanceapprox} is
to use randomization in the selection of $r$ in order to ensure that
$s_r(x)$ satisfies desirable properties in expectation. The key
proposition follows:

\begin{prop}
  Consider strings $x$ and $y$ in $\Sigma^n$. Consider $0 < \epsilon <
  1$ and $R \ge 1$. Select $r$ to be a random real between $R$ and
  $2R$. Then the following three properties hold:
  \begin{itemize}
  \item Every letter $l$ in $s_r(x)$ or $s_r(y)$ satisies $|l| \ge R$.
  \item If $\ed(x, y) \le \frac{nR}{15n^{\epsilon}}$ then $\E[\ed(s_r(x), s_r(y))] \le \frac{nR}{3n^{\epsilon}}$.
  \item If $\ed(x, y) > 5nR$, then $\ed(s_r(x), s_r(y)) > nR$.
  \end{itemize}
  \label{propthreeprops}
\end{prop}

Before proving Proposition \ref{propthreeprops}, we first use it to
prove Theorem \ref{thmeditdistanceapprox}.

\begin{proof}[Proof of Theorem \ref{thmeditdistanceapprox}.]
We begin by defining the \emph{$(R, n^{\epsilon})$-edit-distance gap}
problem for $r \ge 1$, in which for two strings $x$ and $y$ a return
value of 0 indicates that $\ed(x, y) < 5nR$ and a return value of 1
indicates that $\ed(x, y) \ge n^{1 - \epsilon}R / 15$.

One can solve the $(R, n^{1 - \epsilon})$-edit-distance gap problem
with high probability in time $\tilde{O}(n^{2 - \epsilon})$ as
follows. For a sufficiently large number of samples $t = O(\log n)$,
select $r_1, \ldots, r_t$ each at random from $[R, 2R]$. Then apply
Observation \ref{lemkeyeditobservation} to each of the $r_i$'s in
order to obtain an estimate $e_i$ for $\ed(s_{r_i}(x), s_{r_i}(y))$ in
time $O(n^{2 - \epsilon})$. Because each $e_i$ is either correct or an
overestimate, if any of the $e_i$'s are less than $nR$, then we can
conclude that $\ed(s_{r_i}(x), s_{r_i}(y)) < nR$, meaning by Proposition
\ref{propthreeprops} that $\ed(x, y) \le 5nR$. On the other hand, if
none of the $e_i$'s are less than $nR$, then our algorithm for the gap
problem concludes that $\ed(x, y) \ge n^{1 - \epsilon}R / 15$. In
order to show that our algorithm for the gap problem is correct with
high probability, it suffices to show that if $\ed(x, y) < n^{1 -
  \epsilon}R / 15$ then with high probability, at least one of the
$e_i$'s will be less than $nR$. By Proposition \ref{propthreeprops},
if $\ed(x, y) < n^{1 - \epsilon}R / 15$, then $\E[\ed(s_{r_i}(x),
  s_{r_i}(y))] \le \frac{nR}{3n^{\epsilon}}$. Applying Markov's
inequality to each $\ed(s_{r_i}(x), s_{r_i}(y))$, with high
probability there is some $i$ for which
$$\ed(s_{r_i}(x), s_{r_i}(y)) < R n^{1 - \epsilon},$$
which by Observation \ref{lemkeyeditobservation} means that
$$e_i = \ed(s_{r_i}(x), s_{r_i}(y)) < nR,$$
as desired.

Now that we have an algorithm for the $(R, n^{1 -
  \epsilon})$-edit-distance gap problem, the proof of Theorem
\ref{thmeditdistanceapprox} follows just as did the proof of Theorem
\ref{thmdtwapproxtree} using the solution to the $(r,
n^{\epsilon})$-DTW gap problem. Recall, however, that the proof of
Theorem \ref{thmeditdistanceapprox} required the metric space to have
at most exponential aspect ratio. This was so that a binary search
could efficiently find some $i$ for which $(2^{i - 1}, n^{\epsilon} /
2)$-DTW gap returns $1$ and $(2^{i}, n^{\epsilon} / 2)$-DTW gap
returns 0. Such a requirement is not necessary here, because we can
restrict ourselves to considering only the $O(n)$ $i$'s for which
there is at least one letter $l$ in $x$ or $y$ for which $2^{i - 1} \le |l| \le 2^{i + 1}$; indeed, for all other values of $i$, our
algorithm for the $(R, n^{1 - \epsilon})$-edit-distance gap problem
will be guaranteed to return the same value for $R = 2^{i -1}$ as for
$R = 2^i$. By restricting ourselves to these values of $i$, the binary
search to find some $i$ for which $(2^{i - 1}, n^{\epsilon} / 2)$-DTW
gap returns $1$ and $(2^{i}, n^{\epsilon} / 2)$-DTW gap returns 0 can
be performed efficiently regardless of the aspect ratio of the metric
space.
\end{proof}

The rest of the section will be devoted to proving Proposition
\ref{propthreeprops}. Because $r \ge R$, The first part of the
proposition follows immediately from the definition of $s_r(x)$ and
$s_r(y)$.
  
The second part of the proposition comes from the following lemma.

\begin{lem}
    Consider strings $x$ and $y$ in $\Sigma^n$. Consider $R \ge
    1$ and select $r$ to be a random real between $R$ and $2R$. Then
  $$\E[\ed(s_r(x), s_r(y))] \le 5 \ed(x, y).$$
\end{lem}
\begin{proof}
  Consider an optimal sequence $S$ of edits from $x$ to $y$. We will
  consider the cost of simulating this sequence of edits to transform
  $s_r(x)$ to $s_r(y)$. Insertions and deletions are easily simulated
  by either performing the same operation to $s_r(x)$ or performing no
  operation at all (if the operation involves a letter of magnitude
  less than or equal to $r$). Substitutions are slightly more
  complicated as they may originally be between letters $l_1 \in x$
  and $l_2 \in y$ of different magnitudes. By symmetry, we may assume
  without loss of generality that $|l_1| < |l_2|$. We will show that
  the expected cost of simulating the substitution of $l_1$ to $l_2$
  in $s_r(x)$ is at most $5d(l_1, l_2)$. Because insertions and
  deletions can be simulated with no overhead, it follows that
  $\E[\ed(s_r(x), s_r(y))] \le 5\ed(x, y)$.

  If $|l_1| \le r < |l_2|$ then $l_1$ does not appear in $s_r(x)$
  but $l_2$ remains in $s_r(y)$. Thus what was previously a
  substitution of $l_1$ with $l_2$ becomes an insertion of $l_2$ at
  cost $|l_2|$. On the other hand, if we do not have $|l_1| \le r <
  |l_2|$, then either both $l_1$ and $l_2$ are removed from $s_r(x)$
  and $s_r(y)$ respectively, in which the substitution operation no
  longer needs to be performed, or both $l_1$ and $l_2$ are still
  present, in which case the substitution operation can still be
  performed at cost $d(l_1, l_2)$. Therefore, the expected cost of
  simulating the substitution of $l_1$ to $l_2$ in $s_r(x)$ is at most
  \begin{equation}
    \begin{split}
          \Pr[|l_1| \le r < |l_2|] \cdot |l_2| + d(l_1, l_2).
    \end{split}
    \label{eqsimcost}
  \end{equation}

  Because $r$ is selected at random from the range $[R, 2R]$, the
  probability that $|l_1| \le r < |l_2|$ is at most $\frac{|l_2| -
    |l_1|}{R}$. By the triangle inequality, this is at most
  $\frac{d(l_1, l_2)}{R}.$ If we suppose that $|l_2| \le 4R$, then it
  follows by \eqref{eqsimcost} that the expected cost of simulating
  the substitution of $l_1$ to $l_2$ in $s_r(x)$ is at most
  \begin{equation*}
    \begin{split}
          \frac{d(l_1, l_2)}{R} \cdot 4R + d(l_1, l_2) \le 5d(l_1, l_2).
    \end{split}
  \end{equation*}
  If, on the other hand, $|l_2| > 4R$, then in order for $|l_1| \le r$ to be true, we must have $|l_1| \le 2R$, meaning by the
  triangle inequality that $d(l_1, l_2) \ge |l_2| / 2$. Thus in this
  case $|l_2| \le 2d(l_1, l_2)$, meaning by \eqref{eqsimcost} that the
  expected cost of simultating the substitution of $l_1$ to $l_2$ in
  $s_r(x)$ is at most three times as expensive as the original
  substitution.
\end{proof}

Finally the third part of Proposition \ref{propthreeprops} is a
consequence of the following lemma, which completes the proof of
Proposition \ref{propthreeprops}.

\begin{lem}
    Consider strings $x$ and $y$ in $\Sigma^n$. Consider $R \ge
    1$ and select $r$ to be a random real between $R$ and $2R$. Then
    $\ed(x, y) < \ed(s_r(x), s_r(y)) + 4Rn$.
\end{lem}
\begin{proof}
  Beginning with $x$, remove all letters of magnitude at most $r$;
  this costs at most $2nR$. Then obtain $s_r(y)$ through $\ed(s_r(x),
  s_r(y))$ edits. Finally, insert all letters of magnitude at most $r$
  in $y$; this costs at most $2nR$. Combining these three steps, we
  get from $x$ to $y$ at cost $\ed(s_r(x), s_r(y)) + 4nR$.
\end{proof}

\end{document}